\documentclass[aip,jmp,showkeys,showpacs,nofootinbib]{revtex4-1}
\usepackage[utf8]{inputenc}
\usepackage[T1]{fontenc}
\usepackage[english]{babel}
\usepackage{amsmath}
\usepackage{amssymb}
\usepackage{amsfonts}
\usepackage{amsthm} 
\usepackage{natbib}
\usepackage[pdftex]{graphicx}
\usepackage{color}

\allowdisplaybreaks

\newtheorem{theorem}{Theorem}
\newtheorem{lemma}[theorem]{Lemma}
\newtheorem{proposition}[theorem]{Proposition}

\newenvironment{theorem*}[1][Theorem]{\begin{trivlist}\itshape
\item[\hskip \labelsep {\bfseries #1}]}{\normalfont\end{trivlist}}
\newenvironment{remark*}[1][Remark]{\begin{trivlist}
\item[\hskip \labelsep {\bfseries #1}]}{\hfill$\square$\end{trivlist}}
\newenvironment{example*}[1][Example]{\begin{trivlist}
\item[\hskip \labelsep {\bfseries #1}]}{\hfill$\square$\end{trivlist}}

\newcommand{\bigtimes}{\mathop{\mbox{\LARGE$\boldsymbol\times$}}}

\renewcommand{\phi}{\varphi}

\renewcommand{\rho}{\varrho}

\renewcommand{\theta}{\vartheta}

\renewcommand{\d}{\partial}

\newcommand{\vol}{\ensuremath{\mathrm{vol}}}
\newcommand{\eps}{\ensuremath{\varepsilon}}
\newcommand{\ex}{\exists}
\newcommand{\fa}{\forall}
\newcommand{\lnorm}{\left\lVert}
\newcommand{\rnorm}{\right\lVert}
\newcommand{\lbetr}{\left\lvert}
\newcommand{\rbetr}{\right\lvert}
\newcommand{\norm}[1]{\lnorm {#1}\rnorm}

\newcommand{\abs}[1]{\lbetr {#1}\rbetr}

\renewcommand{\l}{\ensuremath{\left}}
\renewcommand{\r}{\ensuremath{\right}}
\newcommand{\sse}{\ensuremath{\subseteq}}

\newcommand{\lin}{\opn{lin}}
\newcommand{\tr}{\ensuremath{\opn{tr}}}
\newcommand{\ubr}{\underbrace}

\newcommand{\oli}{\overline}
\newcommand{\opn}{\operatorname}

\newcommand{\then}{\ensuremath{\Rightarrow}}

\newcommand{\texp}{\ensuremath{\mathrm{texp}}}
\newcommand{\Gf}{\ensuremath{\mathfrak{G}}}
\newcommand{\Hp}{\ensuremath{\mathcal{H}}}
\newcommand{\Dp}{\ensuremath{\mathcal{D}}}
\newcommand{\Sp}{\ensuremath{\mathcal{S}}}
\newcommand{\Lp}{\ensuremath{\mathcal{L}}}
\newcommand{\Wp}{\ensuremath{\mathcal{W}}}
\newcommand{\rn}{\ensuremath{\mathbb{R}}}
\newcommand{\cn}{\ensuremath{\mathbb{C}}}
\newcommand{\nn}{\ensuremath{\mathbb{N}}}
\newcommand{\zn}{\ensuremath{\mathbb{Z}}}
\newcommand{\bra}[1]{\ensuremath{\left\langle{#1}\right|}}
\newcommand{\ket}[1]{\ensuremath{\left|{#1}\right\rangle}}
\newcommand{\disc}{\ensuremath{\mathrm{disc}}}
\newcommand{\clim}{\ensuremath{\mathrm{c\text{-}lim}}}
\newcommand{\id}{\ensuremath{\mathrm{id}}}
\newcommand{\diag}{\ensuremath{\mathrm{diag}}}

\begin{document}

\preprint{DESY 18-139}

\title[Zeta-regularized \mbox{vacuum expectation values}]{Zeta-regularized \mbox{vacuum expectation values}}

\author{T. Hartung}
\affiliation{Department of Mathematics, King's College London, Strand, London WC2R 2LS, United Kingdom}
\author{K. Jansen}
\affiliation{NIC, DESY Zeuthen, Platanenallee 6, 15738 Zeuthen, Germany}

\date{\today}

\begin{abstract}
  It has recently been shown that vacuum expectation values and Feynman path integrals can be regularized using Fourier integral operator $\zeta$-function, yet the physical meaning of these $\zeta$-regularized objects was unknown.

  Here we show that $\zeta$-regularized vacuum expectations appear as continuum limits using a certain discretization scheme. Furthermore, we study the rate of convergence for the discretization scheme using the example of a one-dimensional hydrogen atom in $(-\pi,\pi)$ which we evaluate classically, using the Rigetti Quantum Virtual Machine, and on the Rigetti 8Q quantum chip ``Agave'' device. We also provide the free radiation field as an example for the computation of $\zeta$-regularized vacuum expectation values in a gauge theory.
\end{abstract}


\maketitle

\section{Introduction}
In a quantum field theory (QFT), vacuum expectation values are fundamental objects. As expectation values of observables they allow for experimental verification and to test theoretical models. Furthermore, the Wightman Reconstruction Theorem asserts that a QFT is uniquely determined by its $n$-point functions which are tempered (by the Wightman axioms) distributions whose point evaluations (on functions) are vacuum expectation values of $n$ field operators. Hence, given a QFT with Hilbert space $\Hp$, vacuum state $\psi$, and an operator $A$, we are interested in computing the vacuum expectation value $\langle A\rangle$ of $A$;
\begin{align*}
  \langle A\rangle:=\bra\psi A\ket\psi:=\langle\psi,A\psi\rangle_\Hp.
\end{align*}
In general, we do not have access to $\psi$ but it is possible~\cite{creutz-freedman,feynman,feynman-hibbs-styer} to express $\langle A\rangle$ in terms of operator traces. Let $U=\texp\l(-\frac{i}{\hbar}\int_0^TH(s)ds\r)$ the wave propagator of the QFT where $H$ denotes the Hamiltonian and $\texp$ the time-ordered exponential. Then 
\begin{align*}
  \langle A\rangle=\lim_{T\to\infty}\frac{\tr(UA)}{\tr U}.
\end{align*}
Again, we are in a precarious situation since, in general, neither $U$ nor $UA$ are trace-class operators in $\Hp$. This indicates that the difficulty in defining vacuum expectation values with this approach is the construction of these traces.

  Considering only the partition function $\tr U$, Hawking~\cite{hawking} observed that it is possible to relate $\tr U$ to a $\zeta$-function trace construction for pseudo-differential operators.\cite{kontsevich-vishik,kontsevich-vishik-geometry,ray,ray-singer} Since this is a spectral approach to the trace construction, it requires explicit diagonalization of a second order differential operator which is induced by the background fields and the quadratic term of metric fluctuation. As such, explicit computation of spectrally $\zeta$-regularized partition functions are next to impossible in a non-trivial theory. Furthermore, the approach is not easily extended to the numerator $\tr(UA)$ and the physical meaning of the resulting $\zeta$-regularized partition function remained unknown.

Nevertheless it was shown~\cite{hartung} that both traces, $\tr(UA)$ and $\tr U$, can be constructed in a non-perturbative way using Fourier integral operator $\zeta$-functions. In this formulation, we assume that the Hamiltonian $H$ and the operator $A$ are pseudo-differential operators on a compact Riemannian $C^\infty$-manifold $X$ without boundary (a Cauchy surface of the ``universe''; the infinite volume limit $X\to$ ``non-compact manifold'' is taken after $\zeta$-regularization\cite{hartung} and henceforth ignored for the purposes of this paper) and the Hilbert space $\Hp$ is a Sobolev space $W_2^s(X)$ for some $s\in\rn$. The operators $U$ and $A$ are then bounded linear operators mapping $W_2^s(X)$ to $W_2^{s'}(X)$ for some $s'\in\rn$. It is then necessary to construct a suitable holomorphic family of operators\cite{hartung,hartung-phd,hartung-iwota,hartung-scott} $(\Gf(z))_{z\in\cn}$ which (among a number of technical properties) satisfies two important conditions; namely, $\Gf(0)=1$ and $\ex R\in\rn\ \fa z\in\cn_{\Re(\cdot)<R}:\ U\Gf(z)$ and $U\Gf(z)A$ are of trace-class. Hence, considering the families $U\Gf A$ and $U\Gf$, we can recover the operators we are interested in through point evaluation in zero and for $\Re(z)<R$ the traces $\tr(U\Gf(z)A)$ and $\tr(U\Gf(z))$ are well-defined. In fact, the maps 
\begin{align*}
  \zeta_0(U\Gf A):\ \cn_{\Re\cdot<R}\to\cn;\ z\to\tr(U\Gf(z)A)
\end{align*}
and
\begin{align*}
  \zeta_0(U\Gf):\ \cn_{\Re\cdot<R}\to\cn;\ z\to\tr(U\Gf(z))
\end{align*}
have meromorphic extensions to $\cn$ with at most isolated simple poles. We will denote these extensions by $\zeta(U\Gf A)$ and $\zeta(U\Gf)$ respectively, and we can define the $\zeta$-regularized vacuum expectation value $\langle A\rangle_\Gf$ of $A$ with respect to $\Gf$ as
\begin{align*}
  \langle A\rangle_\Gf:=\lim_{T\to\infty}\frac{\zeta(U\Gf A)}{\zeta(U\Gf)}
\end{align*}
which is meromorphic again. Finally, we are interested in computing $\langle A\rangle_\Gf(0)$ which is ``almost always'' independent\cite{hartung-iwota,hartung} of the choice of $\Gf$ and, in general, \emph{very} difficult to compute. For examples of $\zeta$-regularized vacuum expectation values using Fourier integral operators we refer to references~\cite{hartung,hartung-iwota,hartung-jansen}.

At this stage we therefore have a fully regularized expectation value, but the physical meaning of $\langle A\rangle_\Gf(0)$ is still unclear. Hence, it is precisely the purpose of this paper to provide, for the first time, a physical interpretation of $\langle A\rangle_\Gf(0)$.

\begin{example*}
  This kind of construction of $\Gf$ is called a ``gauge'' in the mathematical literature~\cite{guillemin} and related to many questions in geometric analysis. For instance, let us consider the operator $\abs\d$ on the flat torus $\rn/2\pi\zn$. As we will see later, this is a simplification of the operator considered in Section~\ref{sec:radiation}. Its spectrum is $\sigma(\abs\d)=\{\abs n;\ n\in\zn\}$ counting multiplicities. Hence, we can deduce that $\Gf(z):=\abs\d^z$ has spectrum $\sigma(\Gf(z))=\{\abs n^z;\ n\in\zn\}$ which is absolutely summable, i.e., in $\ell_1(\zn)$, whenever $\Re(z)<-1$.

  Suppose that $\abs \d$ is indeed the Hamiltonian of a quantum field theory we wish to consider. Then $U=\exp\l(-\frac{i}{\hbar}T\abs\d\r)$ is a unitary time evolution operator and $U\Gf(z)$ is of trace-class whenever $\Re(z)<-1$ since all Schatten classes are two-sided ideals in the set of bounded operators on a Hilbert space (here $L_2(\rn/2\pi\zn)$).
\end{example*}

\begin{example*}
  More generally, Radzikowski~\cite{radzikowski-phd,radzikowski} showed that given the Hadamard condition, the operators $U$, $A$, and $\Gf(z)$ can be assumed to be Fourier integral operators. Since the product $U\Gf(z)$ needs to be computed, it is often advantageous to choose $\Gf(z)$ as a function of the Hamiltonian $H$, e.g., $H^z$. This ensures that $U\Gf(z)$ can be computed in terms of a functional calculus. In particular, this ensures that $\Gf(z)$ is a Fourier integral operator of order $\gamma z$ for some $\gamma>0$ if $H$ is unbounded on the Hilbert space. In other words, $\Gf(z)$ is of trace-class whenever $\Re(z)$ is sufficiently negative. For example, if $H$ is a pseudo-differential operator on a manifold $X$, then $\Gf(z)$ is of trace-class provided $\Re(z)<-\frac{\dim X}{\gamma}$.
\end{example*}

\begin{remark*}
  In general, computing these traces through spectral decomposition is very difficult. Instead we use that Fourier integral operators are integral operators whose kernels are much more easily accessible than the spectrum.

  Let $A$ be a pseudo-differential operator with symbol $\sigma(x,y,\xi)$ on an open subset $U$ of $\rn^n$; in other words, $A$ is given by $Au(x)=\int_{\rn^n}\int_Ue^{i\langle x-y,\xi\rangle}\sigma(x,y,\xi)u(y)dyd\xi$. Then, the kernel $k$ of $A$ is $k(x,y)=\int_{\rn^n}e^{i\langle x-y,\xi\rangle}\sigma(x,y,\xi)d\xi$. Thus, if $A$ is a trace-class operator, then\cite{brislawn} $\tr A=\int_Uk(x,x)dx=\int_U\int_{\rn^n}\sigma(x,x,\xi)d\xi dx$. This can be lifted to the manifold case\cite{scott}, that is, if $A$ is locally given by a symbol $\sigma$, then $\tr A$ is given by the density $\int_{\rn^n}\tr\sigma(x,x,\xi)d\xi\abs{dx}$.

  Similarly, if $A$ is a trace-class Fourier integral operator whose (matrix-valued) kernel is locally given by $k(x,y)=\int_{\rn^n}e^{i\vartheta(x,y,\xi)}\sigma(x,y,\xi)d\xi$, then its trace can be computed integrating the density $\int_{\rn^n}\tr(e^{i\vartheta(x,x,\xi)}\sigma(x,x,\xi))d\xi\abs{dx}$.
\end{remark*}

Numerically, we typically access discretized systems which can be understood abstractly by assuming that the Hilbert space is finite dimensional and the operators are matrices. In this sense, we will consider restrictions to finite dimensional subspaces to be discretizations and approximations of the infinite dimensional Hilbert space to be a ``continuum limit''. This therefore includes cases of lattice discretizations but also allows for ``discretization schemes'' which do not formally discretize space-time. For instance, we will consider discretization of Fourier modes in Section~\ref{sec:convergence} which extrapolate to the ``continuum of Fourier modes''.

Defining the traces is thus no problem but we need to compute the continuum limit. Since the regularized traces are defined via meromorphic extension, it is neither obvious that any such continuum limit should exist nor that it coincides with the value $\langle A\rangle_\Gf(0)$.

In this article, we will discuss a method of discretization which ensures that the continuum limit exists and we will prove that it coincides with $\langle A\rangle_\Gf(0)$. Furthermore, the chosen method of discretization is interesting in the context of quantum computing which allows us to approximate discretized vacuum states $\psi_\disc$, e.g., using a variational quantum eigensolver.\cite{vqe}

In particular, we will prove the following Theorem~\ref{thm:continuum-limit} which states that continuum limits are precisely the $\zeta$-regularized vacuum expectation values.
\begin{theorem}\label{thm:continuum-limit}
  Let $\psi_n$ be the vacuum as computed using the discretization scheme $\disc$ (cf. section~\ref{sec:disc}) and $A$ and $\Gf$ such that the sequences $(z\mapsto\norm{\Gf(z)A\psi_n})_{n\in\nn}$ and $(z\mapsto\norm{\Gf(z)\psi_n})_{n\in\nn}$ are locally bounded in $C(\cn)$. Furthermore, let the assumptions of Proposition~\ref{prop:continuum-limit} be satisfied. 

  Then the continuum limit $\clim$ of discretized vacuum expectation values $\langle A_\disc\rangle$ exists and satisfies
  \begin{align*}
    \langle A\rangle_\Gf(0)=&\clim\bra{\psi_n} A\ket{\psi_n}=\clim\langle A_\disc\rangle
    =\clim\bra{\psi_\disc} A_\disc\ket{\psi_\disc}=\bra\psi A\ket\psi.
  \end{align*}
\end{theorem}

From a physical point of view, $\bra\psi A\ket\psi$ is the quantity we would like to compute but we do not have access to it since, in general, we don't know the vacuum $\psi$. On the other hand, $\langle A\rangle_\Gf(0)$ is a mathematically well-defined object applying $\zeta$-regularization to Feynman's path integral. A priori, there is no reason for these two quantities to be related since we have been changing the path integral definition on a very fundamental level. Nonetheless, Theorem~\ref{thm:continuum-limit} states that the two have to coincide, i.e., that physical vacuum expectation values arise as $\zeta$-regularized vacuum expectation values.

Furthermore - and central to proving this statement - both $\bra\psi A\ket\psi$ and $\langle A\rangle_\Gf(0)$ can be expressed as the same continuum limit of $\langle A_\disc\rangle:=\bra{\psi_\disc} A_\disc\ket{\psi_\disc}$ which is the numerical problem of computing $\bra\psi A\ket\psi$ in a certain discretization scheme $\disc$ (further discussed below and in full detail in Section~\ref{sec:disc}). This discretized vacuum expectation can alternatively be stated as $\langle A_\disc\rangle=\bra{\psi_n} A\ket{\psi_n}$ where $\psi_n$ is a discretized approximation to the vacuum $\psi$. Most importantly, $\psi_n$ is accessible on a Quantum Processing Unit for which we choose the Rigetti 8Q chip ``Agave'' in this paper.\cite{rubin} In terms of qubits, $n=2^{\text{number of qubits}}$ and the continuum limit is $n\to\infty$.

However, the most remarkable observations are that, firstly, all computations ($\zeta$~and discretized) are non-perturbative and performed in Minkowski space allowing for real time computations and, secondly, the $\zeta$-computation is in the continuum. This, in combination with quantum computing, can therefore lead to completely new avenues for quantum field theory calculations. Note also that the here described procedure is much more general than standard a lattice theory formulation on a Euclidean space-time grid.

In more mathematical terms, we take the point of view that a discretization scheme is a restriction of a problem posed in a separable Hilbert space $\Hp$ to finite dimensional subspaces. In other words, a discretization scheme is a sequence of projections $(P_n)_{n\in\nn}$ on $\Hp$ such that $\fa n\in\nn:\ \dim(P_n[\Hp])=n$. The discretization scheme $\disc$ will furthermore assume that these projections are nested in the sense $\fa n\in\nn:\ P_n[\Hp]\sse P_{n+1}[\Hp]$. Hence, increasing $n$ can be seen as refinement of the discretization. Such a discretization scheme can, for instance, be constructed using an orthonormal basis $(e_j)_{j\in\nn_0}$ and defining
\begin{align*}
  \fa n\in\nn:\ P_n[\Hp]:=\lin\{e_j;\ 0\le j< n\}.
\end{align*}
Such a construction ensures density of $\bigcup_{n\in\nn}P_n[\Hp]$ in $\Hp$ and makes it fairly easy to compute the matrix $M$ describing the discretization of an operator $A$ on $P_n[\Hp]$. More precisely, 
\begin{align*}
  \fa 0\le j,k<n:\ M_{jk}=\langle e_j,Ae_k\rangle_{\Hp}.
\end{align*}
Given the Hamiltonian $H$ of the system in $\Hp$, the vacuum state $\psi$ is defined to be a normalized minimizer of $x\mapsto\langle x,Hx\rangle_{\Hp}$. Thus, we can obtain approximate vacuum states $\psi_n$ on $P_n[\Hp]$ by minimizing $x\mapsto\langle x,Hx\rangle_\Hp$ over all normalized elements of $P_n[\Hp]$. It should be noted that a priori $\psi_n$ need not be $P_n\psi$, i.e., we do not get $\psi_n\to\psi$ for free. However, since $P_n[\Hp]\sse P_{n+1}[\Hp]$ and each $\psi_n$ minimizes $x\mapsto\langle x,Hx\rangle_\Hp$, we do know that the sequence $(\langle\psi_n,H\psi_n\rangle_\Hp)_{n\in\nn}$ is non-increasing which under additional assumptions (namely those of Proposition~\ref{prop:continuum-limit}) will yield $\psi_n\to\psi$ as well as $\langle\psi_n,A\psi_n\rangle_\Hp\to\langle\psi,A\psi\rangle_\Hp$ for observables $A$. In this sense, the continuum limits in Theorem~\ref{thm:continuum-limit} are to be understood as limits $n\to\infty$.

The remaining assumptions in Theorem~\ref{thm:continuum-limit} for $A$ and $\Gf$ to be such that the sequences $(z\mapsto\norm{\Gf(z)A\psi_n})_{n\in\nn}$ and $(z\mapsto\norm{\Gf(z)\psi_n})_{n\in\nn}$ are locally bounded in $C(\cn)$, are necessary to prove that the limit is precisely the $\zeta$-regularized vacuum expectation. This is due to the fact that - in the proof - we express $\langle A\rangle_\Gf$ as the quotient $\frac{\langle\psi,\Gf(\cdot)A\psi\rangle_\Hp}{\langle\psi,\Gf(\cdot)\psi\rangle_\Hp}$ and then use the discretization scheme on both numerator and denominator separately. Hence, pointwise boundedness of $(z\mapsto\norm{\Gf(z)A\psi_n})_{n\in\nn}$ and $(z\mapsto\norm{\Gf(z)\psi_n})_{n\in\nn}$ is simply one of the assumptions in Proposition~\ref{prop:continuum-limit}. In particular, for $z=0$ they are numerically necessary as otherwise the variance of the observable $A$ is unbounded making any numerical approach to compute the limit unfeasible. However, pointwise boundedness is not quite sufficient for the proof since a pointwise convergent sequence of holomorphic functions might not have a holomorphic limit. On the other hand, we need the limit $\lim_{n\to\infty}\frac{\langle\psi_n,\Gf(\cdot)A\psi_n\rangle_\Hp}{\langle\psi_n,\Gf(\cdot)\psi_n\rangle_\Hp}$ to be meromorphic if we want to conclude that it coincides with $\langle A\rangle_\Gf$.

In order to obtain a more thorough understanding of how the discretization scheme and $\zeta$-regularization work, we will consider two examples before proving Theorem~\ref{thm:continuum-limit} in sections~\ref{sec:disc} and~\ref{sec:proof}. First, we will consider the free radiation field (Section~\ref{sec:radiation}) as an introductory example of a QFT with easy to compute vacuum energy. We will explicitly construct the Hamiltonian $H$ and compute its vacuum energy using both methods, i.e., $\bra0H\ket0$ and $\langle H\rangle_\Gf(0)$. The second example (Section~\ref{sec:convergence}) will be the $1$-dimensional hydrogen atom on $(-\pi,\pi)$ and focus on the the convergence rate in the discretization scheme. This example is chosen in such a way, that the vacuum state is highly non-trivial but the discretizations are numerically easy to handle. In particular, we will use the Rigetti Quantum Virtual Machine and Rigetti 8Q chip ``Agave'' to show that the discretization scheme using the standard Fourier basis converges exponentially fast in the number of qubits and that such a convergence rate can be realized on a quantum computer within the limitations of its fidelity. Thus quantum computation can be a powerful tool to compute vacuum expectation values of observables in Minkowski background even utilizing a small number of qubits.

\section{The free radiation field}\label{sec:radiation}
Before presenting the proof of Theorem~\ref{thm:continuum-limit} we would like to showcase the $\zeta$-regularization part applied to a QFT. This section does not contain the discretization aspect of Theorem~\ref{thm:continuum-limit} but is intended to give a non-trivial example of the continuum computation using the $\zeta$-function definition of vacuum expectation values. In particular, we want to consider a gauge theory which a priori is highly non-trivial to discuss in this context. Further examples including scalar fields and the Dirac field have been reported previously.~\cite{hartung,hartung-iwota,hartung-jansen}

As an example, we will consider the free radiation field, i.e., QED without coupling to matter, on the spatial torus $(\rn/X\zn)^3$. The quantization of the free radiation field using the Gupta Bleuler formalism is well-known in high energy physics. However, in order to make the paper self-contained, we have recapitulated the construction in Appendix~\ref{app:radiation-construction}.

Following the construction in Appendix~\ref{app:radiation-construction}, the Hilbert space $\Hp_{\le N}$ of up to $N$ photons is given by the symmetric tensor product
\begin{align*}
  \Hp_{\le N}=\bigoplus_{n=1}^NS\bigotimes_{j=1}^n\l(L_2\l((\rn/X\zn)^3,\cn^2\r)\ominus\lin\l\{
  \begin{pmatrix}
    0\\1
  \end{pmatrix}
  \r\}\r)\sse L_2\l((\rn/X\zn)^{3N},\cn^2\r)
\end{align*}
and the Hamiltonian $H_N$ is
\begin{align*}
  H_N=\sum_{k=1}^N\l(\bigotimes_{m=1}^{k-1}\id_{L_2\l((\rn/X\zn)^{3},\cn^2\r)}\r)\otimes \abs{\nabla_{L_2\l((\rn/X\zn)^{3},\cn^2\r)}}\id_{\cn^2}\otimes\l(\bigotimes_{m=k+1}^{N}\id_{L_2\l((\rn/X\zn)^{3},\cn^2\r)}\r)
\end{align*}
which extends the example $\abs\d$ on $\rn/2\pi\zn$ as discussed in the introduction.

Since $\abs\nabla$ is a pseudo-differential operator with symbol $\sigma_{\abs\nabla}:\ \rn^3\to\rn;\ \xi\mapsto \norm\xi_{\ell_2(3)}$, we obtain that $H_N$ is pseudo-differential with symbol $\sigma_{H_N}:\ (\rn^3)^{N}\to\rn;\ \xi\mapsto\sum_{n=1}^N\norm{\xi_n}_{\ell_2(3)}\id_{\cn^2}$ and, since $H_N$ is time independent and its kernel independent of the spatial variable, we conclude that the up-to-$N$-photon wave propagator $U_N$ is given by $U_N=e^{-iTH_N}$ and has symbol $\sigma_{U_N}:\ (\rn^3)^{N}\to\rn;\ \xi\mapsto e^{-iT\sum_{n=1}^N\norm{\xi_n}_{\ell_2(3)}}\id_{\cn^2}$. Hence, choosing the gauge $\Gf(z):=H_N^z$, we obtain
\begin{align*}
  \langle H_N\rangle_\Gf(z)=&\lim_{T\to\infty}\frac{\int_{(\rn^3)^N}\tr\l(e^{-iT\sum_{n=1}^N\norm{\xi_n}_{\ell_2(3)}}\sum_{n=1}^N\norm{\xi_n}_{\ell_2(3)}\prod_{m=1}^N\norm{\xi_m}_{\ell_2(3)}^{z}\id_{\cn^2}\r)d\xi}{\int_{(\rn^3)^N}\tr\l(e^{-iT\sum_{n=1}^N\norm{\xi_n}_{\ell_2(3)}}\prod_{m=1}^N\norm{\xi_m}_{\ell_2(3)}^{z}\id_{\cn^2}\r)d\xi}\\
  =&\lim_{T\to\infty}\frac{\sum_{n=1}^N\prod_{m=1}^N\int_{\rn^3}e^{-iT\norm{\xi_m}_{\ell_2(3)}}\norm{\xi_m}_{\ell_2(3)}^{z+\delta_{mn}}d\xi_m}{\prod_{m=1}^N\int_{\rn^3}e^{-iT\norm{\xi_m}_{\ell_2(3)}}\norm{\xi_m}_{\ell_2(3)}^{z}d\xi_m}\\
  =&\lim_{T\to\infty}\frac{N\int_{\rn_{>0}}e^{-iTr}r^{z+3}dr}{\int_{\rn_{>0}}e^{-iTr}r^{z+2}dr}\\
  =&\lim_{T\to\infty}\ubr{\frac{N\Gamma(z+4)(iT)^{-z-4}}{\Gamma(z+3)(iT)^{-z-3}}}_{\propto\frac1T}\\
  =&0.
\end{align*}
These integrals are well-defined in terms of homogeneous distributions whenever $U_NH_N^{z+1}$ is of trace-class which is satisfied for $\Re(z)<-4$.

\begin{remark*}
  Let $\theta$ be a phase function in an open cone $\Gamma\sse X\times \rn^n$, $F\sse \Gamma\cup(X\times\{0\})$ a closed cone, and $a$ in the symbol class $S^\mu_{\rho,\delta}(X\times\rn^n)$ with support in $F$ for some $\mu\in\rn$, $\rho\in(0,1]$, and $\delta\in[0,1)$. Then, $I:\ C_c^\infty(X)\to\cn;\ u\mapsto\int_{X\times\rn^n}e^{i\theta(x,\xi)}a(x,\xi)u(x)d(x,\xi)$ defines a distribution of order $\le k$ whenever $\max\{\mu-k\rho,\mu-k(1-\delta)\}<-n$ (Theorem~7.8.2 in$\ $\cite{hoermander}).

  Hence, in the case of $I:=\int_{\rn^3}e^{-iT\norm{\xi_m}_{\ell_2(3)}}\norm{\xi_m}_{\ell_2(3)}^{z+\delta_{mn}}d\xi_m$, we observe $\rho=1$, $\delta=0$, $\theta(x,\xi_m)=T\norm{\xi_m}_{\ell_2(3)}$, and $a(x,\xi_m)=\norm{\xi_m}_{\ell_2(3)}^{z+\delta_{mn}}$ is a classical symbol with $\mu=\Re(z)+\delta_{m,n}$. In other words, $I$ is a well-defined distribution if $\Re(z)+1<-3$.
\end{remark*}

By construction of the $\zeta$-regularization scheme, the vacuum expectation is then defined via analytic continuation of $\langle H_N\rangle_\Gf(z)$ to $0$, i.e., $\langle H_N\rangle_\Gf(0)=0$. Finally, we can capture all physically reachable states taking the (now trivial) limit $N\to\infty$ and observe
\begin{align*}
  \langle H\rangle=\bra0H\ket0=0=\langle H\rangle_\Gf(0)=\lim_{N\to\infty}\langle H_N\rangle_\Gf(0).
\end{align*}

\begin{remark*}
  The full Fock space can be treated in a similar manner to the computation above. An in-depth discussion of the Fock space case is included in Appendix~\ref{app:fockspace}.
\end{remark*}

The convergence with respect to $\disc$ is not interesting in this case because most reasonable choices of basis vectors naturally contain the vacuum. Hence, choosing any enumeration $(e_n)_{n\in\nn_0}$ of basis vectors will imply that the vacuum $\ket0$ is contained in all $P_n[\Hp]$ with $n$ sufficiently large. This, however, trivializes the limit. We will therefore discuss a non-trivial example from the discretization point of view in Section~\ref{sec:convergence}.

\section{Convergence in \disc}\label{sec:convergence}
Having discussed the $\zeta$-regularization half of Theorem~\ref{thm:continuum-limit}, we want to have a look at the rate of convergence using a non-trivial example. This example illustrates that the discretization scheme can be implemented on a quantum device and, as we will see, convergence rates can be exponential in the number of qubits. This gives access to the discretized vacuum state (minimizing the energy) and hence vacuum expectations can be computed directly using the vacuum state. This can be generalized to include the time evolution and evaluate the path integral. However, this is beyond the scope of this example.

The main advantage of our discretization scheme is the proof of convergence. Thus any discretization scheme used in practice will converge to the physically correct vacuum expectation value provided it satisfies the assumptions of Theorem~\ref{thm:continuum-limit}. However, it is not clear from the construction which rates of convergence to expect. In order to go beyond a mere convergence proof result, we want to explicitly demonstrate that numerically viable convergence rates can be obtained. The application is therefore not intended to devise novel quantum algorithms nor is the implementation competitive with state-of-the-art quantum algorithms. This is based on the fact that the $\zeta$-regularization serves only to provide a rigorous definition of the continuum limit. The construction of the discretization scheme itself is an independent procedure whose limit coincides with the $\zeta$-regularized vacuum expectation value as long as the assumptions of Theorem~\ref{thm:continuum-limit} are satisfied. Whether having an analytically well-defined expression for the continuum limit may lead to novel algorithms or a priori estimates on the convergence is an interesting but independent question.

The example we would like to consider is a $1$-dimensional version of the hydrogen atom in $(-\pi,\pi)$ with the proton sitting in $0$. In other words, the Hamiltonian is given by $H=-\frac{\d^2}{2m}+qU(x)$ where $U$ is the Coulomb potential. However, we will not consider the $3$-dimensional Coulomb potential $U_3(x)=\frac{-1}{\norm x_{\ell_2(3)}}$ since that would be a very severe term (note that $U_3$ is integrable over $(-\pi,\pi)^3$ but not over $(-\pi,\pi)$). Instead we choose the $N$-dimensional Coulomb potential $U_N$ to be the Green's function of the Laplacian on $(-\pi,\pi)^N$. In other words, $U_N(x)=\frac{-1}{\norm x_{\ell_2(N)}^{N-2}}$ for $N\ge3$ and $U_2(x)=\ln\norm x_{\ell_2(2)}$ are well-known, and 
\begin{align*}
  U_1(x)=x\cdot1_{(0,\pi)}(x)=
  \begin{cases}
    x&;\ x\in\l(0,\pi\r)\\
    0&;\ x\in\l(-\pi,0\r]
  \end{cases}
\end{align*}
can be easily verified, since for any $\phi\in C_c^\infty((-\pi,\pi))$ we observe
\begin{align*}
  \int_{-\pi}^\pi (x-y)1_{(0,\pi)}(x-y)\phi''(y)dy=&\int_{-\pi}^x (x-y)\phi''(y)dy
  =\int_{-\pi}^x \phi'(y)dy
  =\phi(x).
\end{align*}
The Hamiltonian $H:\ W_2^2((-\pi,\pi))\sse L_2((-\pi,\pi))\to L_2((-\pi,\pi))$ is then defined to be the second order differential operator\footnote{We note that $H$ is the sum of the unbounded self-adjoint operator $-\frac{\d^2}{2m}$ and the bounded self-adjoint operator $qx1_{(0,\pi)}(x)$. Since $(A+B)^*=A^*+B^*$ whenever at least one of the operators $A$ and $B$ is bounded, we obtain that $H$ is self-adjoint.}
\begin{align*}
  -\frac{\d^2}{2m}+qx1_{(0,\pi)}(x).
\end{align*}
The Hilbert space $\Hp$ is $L_2((-\pi,\pi))$ and the basis of our choice is $\phi_k(x):=\frac{1}{\sqrt{2\pi}}e^{ikx}$ for $k\in\zn$. We will order them as $e_0:=\phi_0$, $e_{2j-1}:=\phi_{-j}$, and $e_{2j}:=\phi_{j}$, i.e., the finite dimensional subspaces the discretization is defined on is given by
\begin{align*}
  \fa n:\ P_n[\Hp]=\lin\l\{\phi_k;\ -\l\lfloor\frac{N}{2}\r\rfloor\le k\le N-1-\l\lfloor\frac{N}{2}\r\rfloor\r\}.
\end{align*}
The matrix elements of the Hamiltonian are then given by
\begin{align*}
  \langle\phi_k,H\phi_k\rangle=&\int_{-\pi}^\pi\frac{k^2}{2m}\frac{1}{2\pi}dx+\int_0^\pi qx\frac{1}{2\pi}dx=\frac{k^2}{2m}+\frac{q\pi}{4}
\end{align*}
and for $k\ne l$
\begin{align*}
  \langle\phi_l,H\phi_k\rangle=&\int_{-\pi}^\pi\frac{k^2}{2m}\frac{e^{i(k-l)x}}{2\pi}dx+\int_0^\pi qx\frac{e^{i(k-l)x}}{2\pi}dx\\
  =&\frac{q\l((-1)^{k-l}(1-i\pi(k-l))-1\r)}{2\pi(k-l)^2}.
\end{align*}
In Figure~\ref{fig:convergence-dimension} we can see the relative truncation error of $\bra{\psi_n}H\ket{\psi_n}$ as a function of $n$. Furthermore, we have added the graph of $n\mapsto 4.85\cdot10^{-7}\cdot e^{-.00644n}$ which indicates exponential convergence of $\l(\bra{\psi_n}H\ket{\psi_n}\r)_{n\in\nn}$.
\begin{figure}
  \includegraphics[scale=.35]{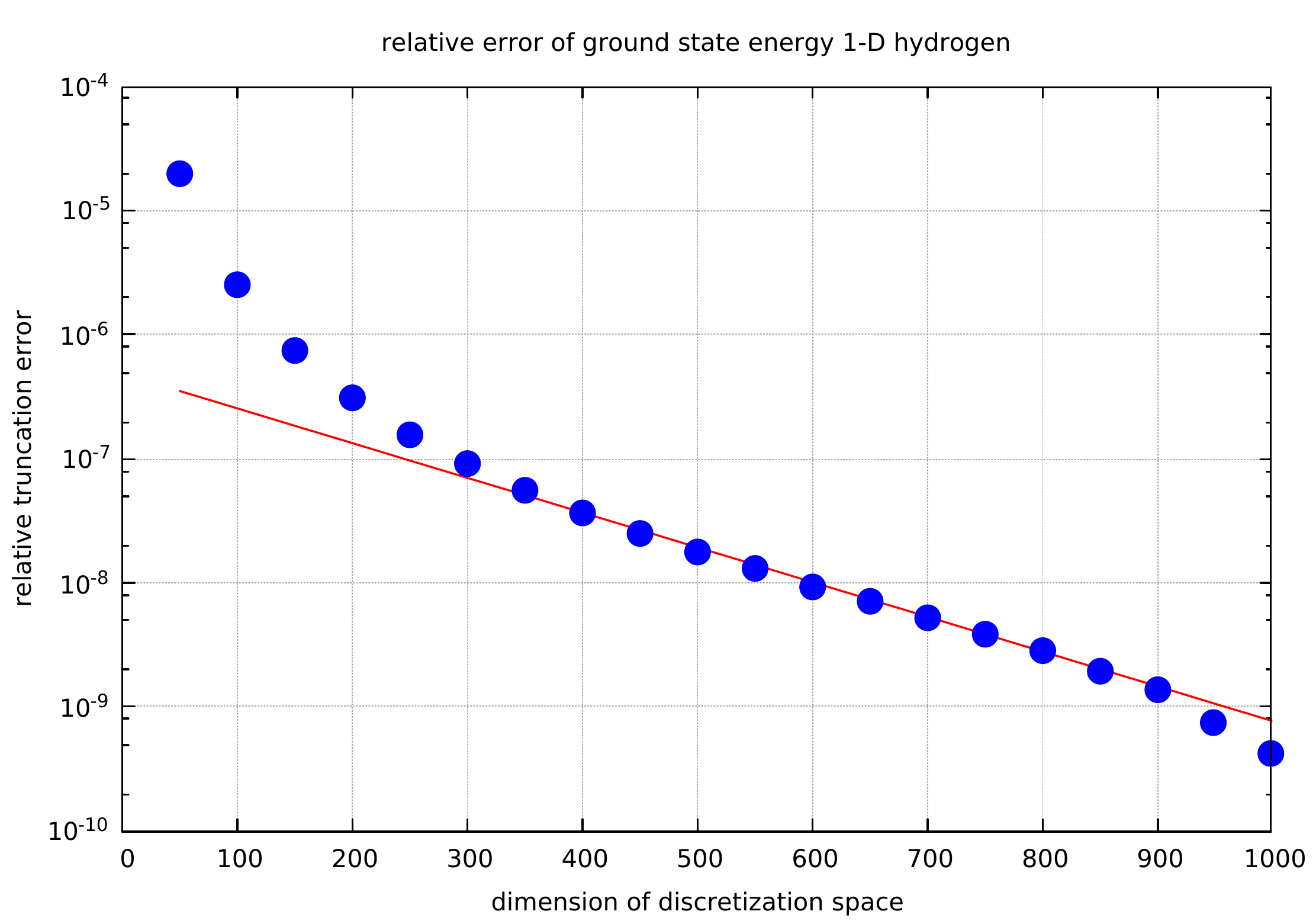}
  \caption{\label{fig:convergence-dimension}This figure shows the relative truncation error of $\bra{\psi_n}H\ket{\psi_n}$ computed for dimension of Hilbert space $n\in\{50,100,\ldots,1000\}$ and compared to $\bra{\psi_{1050}}H\ket{\psi_{1050}}$ as well as the graph of $n\mapsto 4.85\cdot10^{-7}\cdot e^{-.00644n}$ which indicates exponential convergence. We have chosen $m=q=1$ for the mass of the electron $m$ and the electric coupling constant $q$. To compute $\bra{\psi_n}H\ket{\psi_n}$, we have computed the smallest eigenvalue of $H$ restricted to $P_n[\Hp]$.}
\end{figure}

Since the dimension of the Hilbert space of a $q$-qubit quantum computer grows exponentially (more precisely, $n=2^q$), this becomes even more interesting if we consider an implementation on a quantum computer. In that case, Figure~\ref{fig:convergence-qubits} shows the dependence of the relative truncation error of $\bra{\psi_n}H\ket{\psi_n}$ with respect to the number of qubits $q$. It shows that the convergence of $\l(\bra{\psi_{2^q}}H\ket{\psi_{2^q}}\r)_{q\in\nn}$ is comparable to $q\mapsto 1.14\cdot e^{-1.92q}$.
\begin{figure}
  \includegraphics[scale=.35]{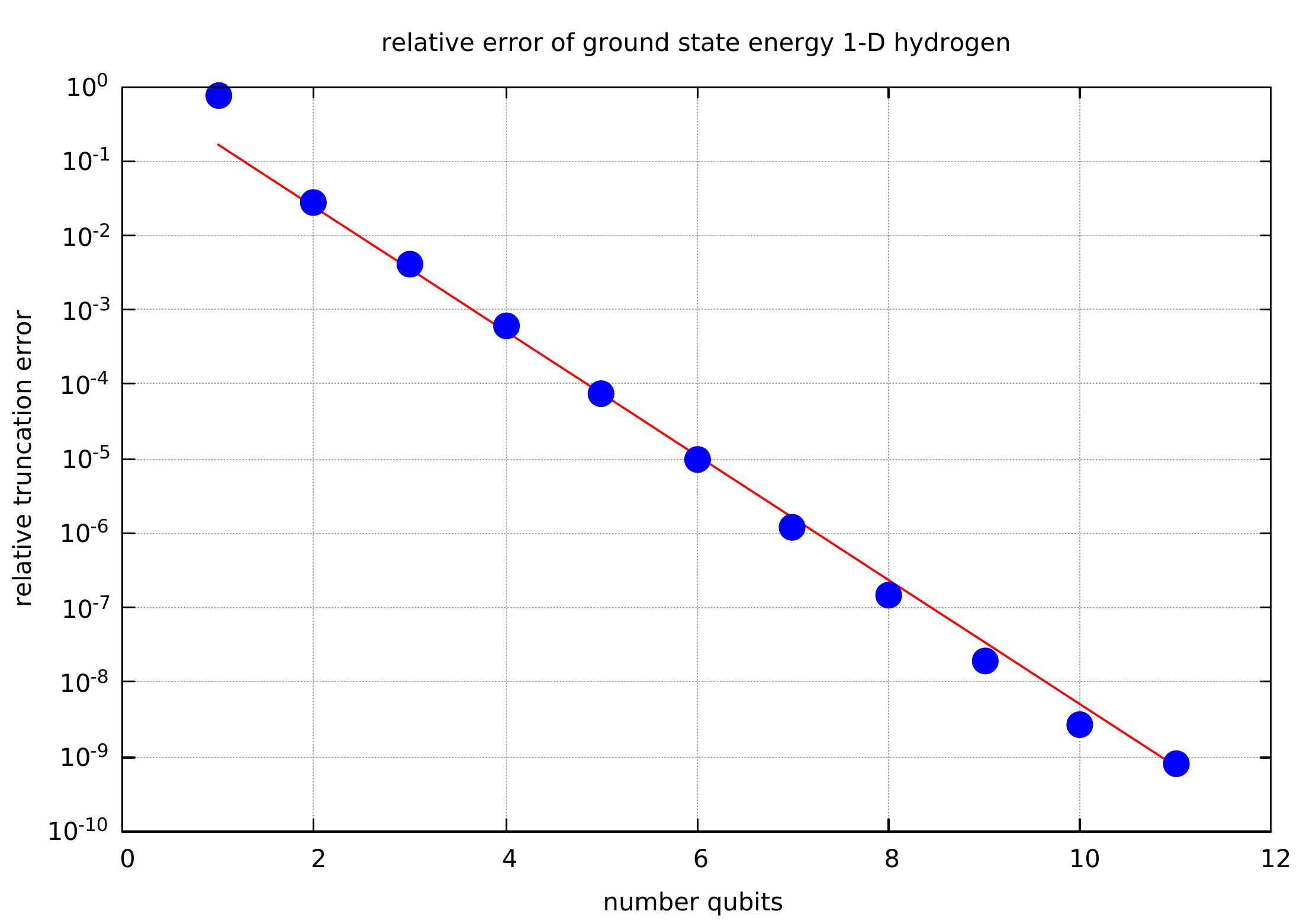}
  \caption{\label{fig:convergence-qubits}This figure shows the relative truncation error of $\bra{\psi_n}H\ket{\psi_n}$ computed for $\mathrm{number\ qubits}\in\{1,2,\ldots,11\}$ and compared to $\bra{\psi_n}H\ket{\psi_n}$ computed with $12$ qubits as well as the graph of $q\mapsto 1.14\cdot e^{-1.92q}$ which indicates exponential convergence. We have chosen $m=q=1$ for the mass of the electron $m$ and the electric coupling constant $q$. To compute $\bra{\psi_n}H\ket{\psi_n}$, we have computed the smallest eigenvalue of $H$ restricted to $P_n[\Hp]$.}
\end{figure}

In order to perform this computation on a quantum computer, we need to map the $Q$-qubit discretized Hamiltonian $H_Q:=\l(\langle e_j,He_k\rangle_{L_2((-\pi,\pi))}\r)_{0\le j,k< 2^Q}$ onto the Pauli-Basis of $\bigotimes_{q=0}^{Q-1}\cn^2$. Following the convention used by Rigetti~\cite{smith}, we will use the Kronecker product notation for tensor products; that is, for vectors
\begin{align*}
  \fa a\in\cn^m\ \fa b\in\cn^n:\ a\otimes b=
  \begin{pmatrix}
    a_0b\\a_1b\\\vdots\\a_{m-1}b
  \end{pmatrix}\in\cn^{mn}
\end{align*}
and for matrices $\fa A\in\cn^{m,n}\ \fa B\in\cn^{r,s}:\ $
\begin{align*}
  A\otimes B = 
  \begin{pmatrix}
    A_{00}B&A_{01}B&\ldots&A_{0n}B\\
    A_{10}B&A_{11}B&\ldots&A_{1n}B\\
    \vdots&\vdots&\ddots&\vdots\\
    A_{m0}B&A_{m1}B&\ldots&A_{mn}B\\
  \end{pmatrix}\in\cn^{mr,ns}.
\end{align*}
In $\cn^2$, we choose the basis $\ket0=(1,0)^T$ and $\ket1=(0,1)^T$. Then we obtain the basis $(\ket q)_{0\le q<2^Q}$ of $\cn^{2^Q}$ of the $Q$-fold tensor product of $\cn^2$s where $\ket q=(\delta_{j,\hat q})_{0\le j< 2^Q}$ with $\hat q\in\nn_0$ having binary representation $q$. For instance, $\ket{100}=(0,0,0,0,1,0,0,0)^T=(\delta_{j,4})_{0\le j<8}$ and $\ket{001}=(0,1,0,0,0,0,0,0)^T=(\delta_{j,1})_{0\le j<8}$.

The Pauli matrices in $\cn^2$ are given by $\sigma^0=\begin{pmatrix}1&0\\0&1\end{pmatrix}$, $\sigma^1=\begin{pmatrix}0&1\\1&0\end{pmatrix}$, $\sigma^2=\begin{pmatrix}0&-i\\i&0\end{pmatrix}$, and $\sigma^3=\begin{pmatrix}1&0\\0&-1\end{pmatrix}$. Using the Kronecker product, we can construct the Pauli basis of $\cn^{2^Q,2^Q}$ which is given by
\begin{align*}
  \l\{S^q=\sigma^{q_{Q-1}}\otimes\sigma^{q_{Q-2}}\otimes\ldots\otimes\sigma^{q_{0}};\ 0\le q< 4^Q\r\}.
\end{align*}
In terms of matrix elements $S^q_{jk}$, let $0\le j,k< 2^Q$ be decomposed as $j=\sum_{n=0}^{Q-1}j_n2^n$ and $k=\sum_{n=0}^{Q-1}k_n2^n$. Then $S^q_{jk}=\prod_{n=0}^{Q-1}\sigma^{q_n}_{j_nk_n}$. The $S^q$ are orthogonal in $\cn^{2^Q,2^Q}$ with respect to the scalar product $(A,B)\mapsto\tr(AB^*)$ and their norm satisfies
\begin{align*}
  \norm{S^q}_{\cn^{2^Q,2^Q}}^2=\tr(S^qS^q)=\tr(\sigma^0\otimes\sigma^0\otimes\ldots\otimes\sigma^0)=2^Q
\end{align*}
which implies
\begin{align*}
  H_Q=\sum_{0\le q< 4^Q}\frac{\tr\l(H_QS^q\r)}{2^Q}S^q.
\end{align*}

Now we are in a position to use the pyQuil variational quantum eigensolver with conjugate gradients in the classical optimization loop. Since each $H_{Q}$ is a restriction of $H_{Q+1}$, we can compute $\l(\langle\psi_{2^Q},H_Q\psi_{2^Q}\rangle\r)_{Q\in\nn}$ iteratively starting with $Q=1$ which is relatively cheap as it is only a minimization problem in a three dimensional parameter space and then use $\psi_{2^Q}$ as initial guess for the minimization problem $\cn^{2^{Q+1}}\ni\psi\mapsto\langle\psi,H_{Q+1}\psi\rangle\to\min$. For $Q\in\{1,2,3,4,5\}$ we then obtain Table~\ref{table:QVM} on the Rigetti Quantum Virtual Machine ignoring noise terms which shows that an implementation on a quantum computer with sufficiently high fidelity can realize the convergence shown in Figure~\ref{fig:convergence-qubits}. 
\begin{table}
  \begin{center}
    \begin{tabular}{|c|c|c|c|}
      \hline
      $Q$ qubits&min. eig. $H_Q$&$\langle\psi_{2^Q},H_Q\psi_{2^Q}\rangle$&$\langle H_Q\rangle$-min. eig.\\
      \hline 1&.392108816647&.392108816647&0.0\\
      \hline 2&.229395425745&.229395425968&2.22839913189e-10\\
      \hline 3&.224258841712&.224258841747&3.48265860595e-11\\
      \hline 4&.223452200306&.223452200445&1.39043221381e-10\\
      \hline 5&.223336689755&.223336690423&6.67360250395e-10\\
      \hline 
    \end{tabular}
  \end{center}
  \caption{\label{table:QVM} Simulation of the ground state energy computation for the $1$-dimensional hydrogen atom using the Rigetti Quantum Virtual Machine (ignoring noise) and comparison to minimal eigenvalue of $H_Q$. We note that the error terms presented in this table are not with respect to the continuum limit but the errors we obtained at each corresponding point in Figure~\ref{fig:convergence-qubits} after replacing the direct computation of the minimal eigenvalue with the QVM-VQE results. Hence, the simulation on the Rigetti Quantum Virtual Machine reproduces Figure~\ref{fig:convergence-qubits}.}
\end{table}

In fact, direct implementation on the Rigetti 8Q chip ``Agave'' in pyQuil~\cite{quil,rigetti} through the Rigetti Forest API yields comparable results within the limitations of the chip's fidelity. Here, we performed a simple, sequential loop over all parameters of the resource Hamiltonian and thus minimizing the energy one parameter at a time where we looped thrice over the set of all parameters. Using a single qubit, the state preparation requires only five gate operations which keeps gate noise low. ``Agave's'' fidelity for single qubit operations (at the time of execution) was best on qubit $2$ which benchmarked a single-qubit gate fidelity $F_{1Q}=0.982$ and readout fidelity $F_{RO}=0.94$. Performing the computation on ``Agave's'' qubit $2$, we obtained the one qubit ground state energy with a relative mean error of $4.9\%$ at $2.8\%$ standard deviation of the relative mean error. We have repeated the computation on ``Agave's'' qubit $0$ whose fidelity benchmarked at $F_{1Q}=0.956$ and $F_{RO}=0.78$. There, we only reached a relative mean error of approximately $15\%$ which highlights the great impact fidelity has.

For progression to the two qubit computation, we thus chose the two qubits with the best fidelity. Unfortunately, we could not achieve significant results in this setup. Nonetheless, high precision results on many qubits were not the primary goal of our work nor to be expected since this example of a $1$-dimensional hydrogen atom in a bounded universe was chosen to ensure that the vacuum is highly non-trivial from the point of view of the discretization scheme. In particular, the wave function in the continuum is unknown to us and we are only able to obtain comparisons with classical algorithms because the low number of qubits still permits numerical diagonalization of the Hamiltonian. However, scaling the computation to $50$ for instance, direct computation of the minimal eigenvalue and corresponding eigenvector (for a $2^{50}\times2^{50}$ matrix) are practically impossible and Monte-Carlo algorithms are not applicable since we are working in a Minkowski background. Yet the noise-free QVM results and the comparison between ``Agave's'' qubits $0$ and $2$ show that improved quantum processing units will provide numerical access to expectation values and vacuum wave functions in physical simulations of this kind.

\begin{remark*}
  The example discussed in this section is purposely chosen to be very difficult to solve. In general, all $4^{\#\text{ Qubits}}$ Pauli terms will need to be measured in order to evaluate the Hamiltonian. Given this exponential scaling behavior, we did not attempt an optimal implementation on the quantum device nor complete complexity and resource analysis. However, the fact that we still observe exponential convergence is very promising for simpler models in which quantum algorithms with polynomial scaling are known.
\end{remark*}

\section{The discretization scheme $\disc$}\label{sec:disc}
At this point, we want to commence the proof of Theorem~\ref{thm:continuum-limit}. The proof will be in two steps. First, we will discuss the discretization in detail and prove its properties. This will mainly prove the part of Theorem~\ref{thm:continuum-limit} concerning
\begin{align*}
  \bra\psi A\ket\psi=&\clim\bra{\psi_\disc} A_\disc\ket{\psi_\disc}
  =\clim\bra{\psi_n} A\ket{\psi_n}
  =\clim\langle A_\disc\rangle.
\end{align*}
In a second step (Section~\ref{sec:proof}) we will then address the $\zeta$-regularized part of Theorem~\ref{thm:continuum-limit}.

For the discretization, we need to approximate our operators using matrices. In other words, we need to project onto finite dimensional spaces. Since we have a holomorphic family of operators $\Gf$, it is imperative that these projections make sense for every value of $z$. This is possible, since by construction of $\Gf$ all of our operators are well-defined on $W_2^\infty(X)$ (or $W_2^s(X)$ for some $s\in\rn$ provided we introduce an upper bound on $\Re(z)$).

Let $\Hp_0\sse\Hp$ be a dense subspace and $(e_j)_{j\in\nn_0}$ an orthonormal basis of $\Hp$ with $\fa j\in\nn_0:\ e_j\in\Hp_0$. For $n\in\nn$, we define the orthogonal projection $P_n$ onto the $n$-dimensional subspace $P_n[\Hp]=\lin\{e_j;\ 0\le j< n\}\sse\Hp_0$
\begin{align*}
  P_n:\ \Hp\to\Hp;\ \phi\mapsto\sum_{j=0}^{n-1}\langle e_j,\phi\rangle_\Hp e_j.
\end{align*}
For any given upper bound $R_\Re\in\rn_{>0}$ on $\Re(z)$, we choose $\Hp_0$ such that we can find a continuously embedded Hilbert space $\Hp_1$ with $\Hp_0\sse\Hp_1\sse\Hp$ for which all $\Gf(z)$ and $\Gf(z)A$ with $\Re(z)<R_\Re$ are in $L(\Hp_1,\Hp)$. Then, we can discretize $\Gf(z)$ and $\Gf(z)A$ as $P_n\Gf(z)Q_n$, and $P_n\Gf(z)AQ_n$ where $Q_n$ is the orthogonal projection onto $P_n[\Hp]$ in $\Hp_1$ and the discretization of $U$ is given in terms of the discretized Hamiltonian $P_nHQ_n$, that is, $U_\disc=\texp\l(-\frac{i}\hbar\int_0^T P_nH(s)Q_nds\r)$.

\begin{example*}
  For instance, let $\Hp$ be $L_2(X)$ and $A$ a pseudo-differential operator of order $\alpha>0$. Then, $\Gf(z)$ can be constructed to be a pseudo-differential operator of order $z$, we can choose $\Hp_0$ to be contained in $W_2^\infty(X)$, and define $\Hp_1:=W_2^s(X)$ for any $s>\alpha+R_\Re$. Thus, all $\Gf(z)$ and $\Gf(z)A$ are pseudo-differential operators of order $\le\alpha+R_\Re$ for $\Re(z)<R_\Re$ and therefore elements of $L(\Hp_1,\Hp)$.

  More generally, the assumptions on the class of operators~\cite{hartung} imply that we are working in algebras of Fourier integral operators with finite order. This setup always allows for a construction using $\Hp_0\sse W_2^\infty$ and $\Hp_1=W_2^s$. 
\end{example*}

This discretization is viable in the sense of the following lemmas and particularly Proposition~\ref{prop:continuum-limit} which says that the discretized Hamiltonian is still self-adjoint, the ground state energy computed in $\disc$ converges to the ground state energy of the continuum (in $\Hp$), and the vacuum computed in $\disc$ ``converges'' to the vacuum in $\Hp$. Finally, Proposition~\ref{prop:continuum-limit} states that $\bra{\psi_\disc}A_\disc\ket{\psi_\disc}$ converges to the vacuum expectation $\bra\psi A\ket\psi$ given any operator $A\in L(\Hp_1,\Hp)$ for which the vacuum of $\Hp$ is in the domain of $A^*$ where $A^*$ is the adjoint of $A$ as an unbounded operator in $\Hp$ (this is the case for pseudo-differential $A$ because $A^*$ is a pseudo-differential operator of the same order) and such that $\norm{A\psi_\disc}_\Hp=\sqrt{\bra{\psi_\disc}(A^*A)_\disc\ket{\psi_\disc}}$ is bounded. While the last assumption - boundedness of the variance of the discretized operator $A$ in the continuum limit - is technical, it is essentially necessary for numerical applications since otherwise the vacuum expectations $\langle P_nAQ_n\rangle$ in $P_n[\Hp]$ are virtually impossible to compute numerically for large $n$.
\begin{lemma}\label{lemma:pointwise_1}
  Let $A\in L(\Hp_1,\Hp)$, $(e_j')_{j\in\nn_0}$ an orthonormal basis of $\Hp_1$ such that $\fa n\in\nn:\ \lin\{e_j';\ 0\le j<n\}=P_n[\Hp]$, and $Q_n$ the orthogonal projection onto $P_n[\Hp]$ in $\Hp_1$. Then $P_nAQ_n$ converges to $A$ in the strong operator topology, i.e.,
  \begin{align*}
    \fa x\in \Hp_1:\ \lim_{n\to\infty}P_nAQ_nx=Ax.
  \end{align*}
\end{lemma}

\begin{proof}
  Let $x\in\Hp_1$. Then
    \begin{align*}
      \norm{Ax-P_nAQ_nx}_\Hp^2
      =&\sum_{j=0}^{n-1}\abs{\l\langle e_j,Ax-P_nAQ_nx\r\rangle_\Hp}^2+\sum_{j\in\nn_{\ge n}}\abs{\l\langle e_j,Ax-P_nAQ_nx\r\rangle_\Hp}^2\\
      =&\sum_{j=0}^{n-1}\abs{\l\langle e_j,A(1-Q_n)x\r\rangle_\Hp}^2
      +\sum_{j\in\nn_{\ge n}}\abs{\l\langle e_j,Ax\r\rangle_\Hp}^2\\
      \le&\norm{A(1-Q_n)x}_\Hp^2+\sum_{j\in\nn_{\ge n}}\abs{\l\langle e_j,Ax\r\rangle_\Hp}^2\\
      \le&\norm A_{L(\Hp_1,\Hp)}^2\ubr{\sum_{j\in\nn_{\ge n}}\abs{\l\langle e_j',x\r\rangle_{\Hp_1}}^2}_{\to0}+\ubr{\sum_{j\in\nn_{\ge n}}\abs{\l\langle e_j,Ax\r\rangle_\Hp}^2}_{\to0}
    \end{align*}
  shows the assertion.
\end{proof}

\begin{lemma}\label{lemma:pointwise_2}
  Let $B\in L(\Hp)$, $A\in L(\Hp_1,\Hp)$, $(e_j')_{j\in\nn_0}$ an orthonormal basis of $\Hp_1$ such that $\fa n\in\nn:\ \lin\{e_j';\ 0\le j<n\}=P_n[\Hp]$, and $Q_n$ the orthogonal projection onto $P_n[\Hp]$ in $\Hp_1$. Then 
  \begin{align*}
    \fa x\in \Hp_1:\ \lim_{n\to\infty}P_nBQ_nP_nAQ_nx=BAx.
  \end{align*}
\end{lemma}

\begin{proof}
  Let $x\in\Hp_1$, $B_n:=P_nBQ_n$, and $A_n:=P_nAQ_n$. Then 
    \begin{align*}
      \norm{BAx-B_nA_nx}_\Hp^2
      =&\sum_{j=0}^{n-1}\abs{\l\langle e_j,BAx-B_nA_nx\r\rangle_\Hp}^2+\sum_{j\in\nn_{\ge n}}\abs{\l\langle e_j,BAx-B_nA_nx\r\rangle_\Hp}^2\\
      =&\sum_{j=0}^{n-1}\abs{\l\langle e_j,(BA-BP_nAQ_n)x\r\rangle_\Hp}^2+\sum_{j\in\nn_{\ge n}}\abs{\l\langle e_j,BAx\r\rangle_\Hp}^2\\
      \le&\norm{B(A-A_n)x}_\Hp^2+\sum_{j\in\nn_{\ge n}}\abs{\l\langle e_j,BAx\r\rangle_\Hp}^2\\
      \le&\norm B_{L(\Hp)}^2\ubr{\norm{(A-A_n)x}_\Hp^2}_{\to0}+\ubr{\sum_{j\in\nn_{\ge n}}\abs{\l\langle e_j,BAx\r\rangle_\Hp}^2}_{\to0}
    \end{align*}
  shows the assertion.
\end{proof}

\begin{proposition}\label{prop:continuum-limit}
  Let $H$ be the Hamiltonian, i.e., $\fa s\in[0,T]:\ H(s)$ is a self-adjoint operator, each $-H(s)$ generates a $C_0$-semigroup, and there exists $E_0:=\min\sigma(H(s))$ such that $E_0$ is in the point spectrum, $\ker (H(s)-E_0)$ is independent of $s$, and $\ex\eps\in\rn_{>0}\ \fa s\in[0,T]:\ B(E_0,\eps)\cap\sigma(H(s))=\{E_0\}$ (in other words, the QFT has an energy gap). Then the following are true.
  \begin{enumerate}
  \item[(i)] Let $A$ be self-adjoint in $\Hp$. Then $P_nAQ_n$ is self-adjoint on $(P_n[\Hp],\langle\cdot,\cdot\rangle_\Hp)$. In particular, $P_nH(s)Q_n$ is self-adjoint.

  \item[(ii)] Let $\psi$ be the vacuum state (i.e., an eigenvector with respect to $E_0$) and $\psi_n$ the vacuum state of $P_nH(s)Q_n$ in $P_n[\Hp]$. Then
    \begin{align*}
      \lim_{n\to\infty}\langle\psi_n,P_nH(s)Q_n\psi_n\rangle_\Hp=&\lim_{n\to\infty}\langle\psi,P_nH(s)Q_n\psi\rangle_\Hp
      =\langle\psi,H(s)\psi\rangle_\Hp=E_0.
    \end{align*}
  \item[(iii)] Let $\hat H:=H(s)-E_0$. If the vacuum is non-degenerate, i.e., $\ker\hat H=\lin\{\psi\}$, then $\langle\psi_n,\psi\rangle_{\Hp}\psi_n\to\psi$ in $\Hp$.

  \item[(iv)] Let the vacuum be non-degenerate and $A\in L(\Hp_1,\Hp)$ be such that the sequence $(\norm{A\psi_n}_\Hp)_{n\in\nn}$ is bounded and $\psi\in D(A^*)$ where $A^*$ is the adjoint of $A$ as an unbounded operator in $\Hp$. Then 
    \begin{align*}
      \langle\psi_n,A\psi_n\rangle_{\Hp}=\langle\psi_n,P_nAQ_n\psi_n\rangle_\Hp\to\langle\psi,A\psi\rangle_\Hp
    \end{align*}
    and
    \begin{align*}
      \langle\psi,P_nAQ_n\psi\rangle_\Hp\to\langle\psi,A\psi\rangle_\Hp.
    \end{align*}
  \end{enumerate}
\end{proposition}
\begin{proof}
  ``(i)'' Since $(P_n[\Hp],\langle\cdot,\cdot\rangle_\Hp)$ is a finite dimensional complex Hilbert space, we know that $P_nAQ_n$ is self-adjoint if and only if its numerical range 
  \begin{align*}
    \mathrm{NR}(P_nAQ_n):=\{\langle\phi,P_nAQ_n\phi\rangle_\Hp;\ \phi\in P_n[\Hp], \norm\phi_\Hp=1\}
  \end{align*}
  is an interval. By the Hausdorff-Toeplitz theorem, $\mathrm{NR}(P_nAQ_n)$ is convex and compact. Hence, it suffices to show that $\mathrm{NR}(P_nAQ_n)\sse\rn$. However, that claim follows directly from self-adjointness of $A$ in $\Hp$ 
  \begin{align*}
    \fa\phi\in P_n[\Hp]:\ \langle\phi,P_nAQ_n\phi\rangle_\Hp=&\langle P_n\phi,AQ_n\phi\rangle_\Hp=\langle\phi,A\phi\rangle_{\Hp}\in\rn
  \end{align*}
  since $P_n$ and $Q_n$ are the identity on $P_n[\Hp]$.

  ``(ii)'' In order to show the convergence claim, we will first note that 
  \begin{align*}
    \lim_{n\to\infty}\langle\psi,P_nH(s)Q_n\psi\rangle_\Hp=\langle\psi,H(s)\psi\rangle_\Hp
  \end{align*}
  follows directly from $\norm{(H(s)-P_nH(s)Q_n)\psi}_\Hp\to0$.

  Regarding $\l(\langle\psi_n,P_nH(s)Q_n\psi_n\rangle_\Hp\r)_{n\in\nn}$, we note that $\fa m,n\in\nn:\ m\ge n\ \then P_n[\Hp]\sse P_m[\Hp]$ implies that $\l(\langle\psi_n,P_nH(s)Q_n\psi_n\rangle_\Hp\r)_{n\in\nn}$ is non-increasing. Furthermore, since $H(s)$ is self-adjoint, its spectrum coincides with its approximate point spectrum which itself is contained in the closure of the numerical range of an operator, i.e., we obtain 
  \begin{align*}
    E_0=&\min\mathrm{NR}(H(s))=\min\{\langle\phi,H(s)\phi\rangle_\Hp;\ \phi\in D(H(s)),\ \norm\phi_\Hp=1\}
  \end{align*}
  since $H-E_0\ge0$ implies $\langle\phi,H(s)\phi\rangle_\Hp\ge E_0$. Hence, $\l(\langle\psi_n,P_nH(s)Q_n\psi_n\rangle_\Hp\r)_{n\in\nn}$ is convergent to some value $E\ge E_0=\langle\psi,H(s)\psi\rangle_\Hp$. 

  Since $\fa x\in P_n[\Hp]\setminus\{0\}:\ E\le\langle\psi_n,P_nH(s)Q_n\psi_n\rangle_\Hp=\langle\psi_n,H(s)\psi_n\rangle_\Hp\le\frac{\langle x,H(s)x\rangle_\Hp}{\norm x_\Hp^{2}}$, it suffices to find a sequence $(x_n)_{n\in\nn}$ such that $\fa n\in\nn\ \ex m\in\nn:\ x_n\in P_m[\Hp]$, $\norm{x_n-\psi}_{\Hp}\to0$, and $\norm{H(s)x_n-H(s)\psi}_\Hp\to0$ as these conditions imply
  \begin{align*}
    E_0\le&E\le\langle x_n,H(s)x_n\rangle_\Hp=\langle x_n-\psi,H(s)x_n\rangle_\Hp+\langle \psi,H(s)x_n\rangle_\Hp\to E_0
  \end{align*}
  using $\abs{\langle x_n-\psi,H(s)x_n\rangle_\Hp}\le\ubr{\norm{x_n-\psi}_\Hp}_{\to0}\ubr{\norm{H(s)x_n}_\Hp}_{\text{bounded}}\to0$ as well as $\langle \psi,H(s)x_n\rangle_\Hp\to\langle \psi,H(s)\psi\rangle_\Hp$.

  In order to find such a sequence, we will use that each $-H(s)$ generates a $C_0$-semigroup, i.e., there exists a strictly increasing $(\lambda_k)_{k\in\nn}\in\l(\rho(-H(s))\cap\rn_{>0}\r)^{\nn}$ such that $\lambda_k\nearrow\infty$, $\lambda_k(\lambda_k+H(s))^{-1}\psi\to\psi$ in $\Hp$, and $H(s)\lambda_k(\lambda_k+H(s))^{-1}\psi\to H(s)\psi$ in $\Hp$ where the Yosida approximation $H(s)\lambda_k(\lambda_k+H(s))^{-1}=\lambda_k-\lambda_k^2(\lambda_k+H(s))^{-1}$ of $H(s)$ is a bounded operator on $\Hp$.

  Choosing one such sequence $(\lambda_k)_{k\in\nn}$, let $y_k:=\lambda_k(\lambda_k+H(s))^{-1}\psi\in\Hp_1$. Then we obtain that 
  \begin{align*}
    \fa k\in\nn:\ &\norm{Q_my_k-y_k}_{\Hp_1}\to0\ (m\to\infty)\\
    &\norm{Q_my_k-y_k}_{\Hp}\le\norm{\id}_{L(\Hp_1,\Hp)}\norm{Q_my_k-y_k}_{\Hp_1}\to0
  \end{align*}
  and
  \begin{align*}
    \norm{H(s)Q_my_k-H(s)y_k}_{\Hp}
    \le&\norm{H(s)}_{L(\Hp_1,\Hp)}\norm{Q_my_k-y_k}_{\Hp_1}
    \to0\ (m\to\infty).
  \end{align*}
  For $n\in\nn$ choose $k_n,m_n\in\nn$ such that $\norm{y_{k_n}-\psi}_\Hp<\frac{1}{2n}$, $\norm{Q_{m_n}y_{k_n}-y_{k_n}}_{\Hp}<\frac{1}{2n}$, $\norm{H(s)y_{k_n}-H(s)\psi}_\Hp<\frac{1}{2n}$, and $\norm{H(s)Q_{m_n}y_{k_n}-H(s)y_{k_n}}_{\Hp}<\frac{1}{2n}$. Then 
  \begin{align*}
    \fa n\in\nn:\ x_n:=Q_{m_n}y_{k_n}\in P_{m_n}[\Hp]
  \end{align*}
  implies $x_n\to\psi$ and $H(s)x_n\to H(s)\psi$ in $\Hp$ and thus the assertion.

  ``(iii)'' Since $\hat H$ is strictly positive on the orthocomplement of $\ker \hat H$, there exists $\eps\in\rn_{>0}$ such that $\fa x\in(\ker \hat H)^\perp:\ \langle x,\hat Hx\rangle_{\Hp}\ge\eps\norm x_\Hp^2$. Let $\pi$ be the orthoprojector onto $\ker \hat H$. Then we observe
  \begin{align*}
    0=&\lim_{n\to\infty}\langle\psi_n,\hat H\psi_n\rangle_{\Hp}
    =\lim_{n\to\infty}\ubr{\langle\pi\psi_n,\hat H\pi\psi_n\rangle_{\Hp}}_{=0}+\langle(1-\pi)\psi_n,\hat H(1-\pi)\psi_n\rangle_{\Hp}\\
    \ge&\lim_{n\to\infty}\eps\norm{(1-\pi)\psi_n}_\Hp^2.
  \end{align*}
  Hence, $\norm{(1-\pi)\psi_n}_\Hp\to0$ and 
  \begin{align*}
    1=\norm{\psi_n}_\Hp^2=\norm{\pi\psi_n}_\Hp^2+\norm{(1-\pi)\psi_n}_\Hp^2
  \end{align*}
  implies $\abs{\langle\psi_n,\psi\rangle_\Hp}=\norm{\pi\psi_n}_\Hp\to1$ and
  \begin{align*}
    \norm{\langle\psi_n,\psi\rangle_\Hp\psi_n-\psi}_\Hp^2
    =&\norm{\langle\psi_n,\psi\rangle_\Hp\pi\psi_n-\psi+\langle\psi_n,\psi\rangle_\Hp(1-\pi)\psi_n}_\Hp^2\\
    =&\norm{\abs{\langle\psi_n,\psi\rangle_\Hp}^2\psi-\psi}_\Hp^2+\abs{\langle\psi_n,\psi\rangle_\Hp}^2\norm{(1-\pi)\psi_n}_\Hp^2\\
    =&\ubr{\abs{\abs{\langle\psi_n,\psi\rangle_\Hp}^2-1}}_{\to0}+\ubr{\abs{\langle\psi_n,\psi\rangle_\Hp}^2}_{\to1}\ubr{\norm{(1-\pi)\psi_n}_\Hp^2}_{\to0}
  \end{align*}
  shows the assertion.

  ``(iv)'' Finally, (iv) follows directly from (iii) since we can assume $\psi_n=\frac{\langle\psi_n,\psi\rangle_\Hp}{\abs{\langle\psi_n,\psi\rangle_\Hp}}\psi_n$ without loss of generality and observe
  \begin{align*}
    \abs{\langle\psi_n,A\psi_n\rangle_\Hp-\langle\psi,A\psi\rangle_\Hp}
    =&\abs{\langle\psi_n-\psi,A\psi_n\rangle_\Hp-\langle\psi,A(\psi_n-\psi)\rangle_\Hp}\\
    \le&\ubr{\norm{\psi_n-\psi}_\Hp}_{\to0}\ubr{\norm{A\psi_n}_\Hp}_{\text{bounded}}+\norm{A^*\psi}_\Hp\ubr{\norm{\psi_n-\psi}_\Hp}_{\to0}
  \end{align*}
  as well as
  \begin{align*}
    \langle\psi,P_nAQ_n\psi\rangle_\Hp\to\langle\psi,A\psi\rangle_\Hp
  \end{align*}
  which follows from Lemma~\ref{lemma:pointwise_1}.
\end{proof}

\begin{remark*}
  The non-degeneracy assumption on the vacuum has to be satisfied in a Wightman theory. In a free field theory, this is satisfied because the theory is essentially an infinite-dimensional harmonic oscillator and its ground state is essentially an infinite tensor product of one-dimensional harmonic oscillator ground states. Similarly, in condensed matter physics, most phases of a material have a unique ground state as well. In the case of moduli spaces, high degeneracy is generally possible but they usually generate their own superselection sectors which are separated and hence the Hilbert space $\Hp$ is restricted to be one such superselection sector making the vacuum in $\Hp$ unique again.
\end{remark*}

Since $\Hp$ and $\Hp_1$ are often Sobolev space $W_2^s(X)$, it is canonically possible to choose $\Hp_1$ in such a way, that the operators $H$, $A$, $\ldots$ are of trace-class, i.e., in the Schatten class $\Sp_1(\Hp_1,\Hp)$. In that case, we can improve Lemma~\ref{lemma:pointwise_1}. 
\begin{lemma}\label{lemma:convergence_in_Sp}
  Let $p\in\rn_{\ge1}\cup\{\infty\}$ and $A\in\Sp_p(\Hp_1,\Hp)$. Then $\norm{P_nAQ_n-A}_{\Sp_p(\Hp_1,\Hp)}\to0$ and hence $\norm{P_nAQ_n-A}_{L(\Hp_1,\Hp)}\to0$ as well. 
\end{lemma}
\begin{proof}
  Let $\eps\in\rn_{>0}$. There exists a finite rank operator $A_0=\sum_{k=0}^{K-1}\alpha_k\langle x_k,\cdot\rangle_{\Hp_1}y_k$ with $\alpha\in\cn^K$ and orthonormal families $(x_k)_{0\le k< K}\in\Hp_1^K$ and $(y_k)_{0\le k< K}\in\Hp^K$ such that $\norm{A-A_0}_{\Sp_p(\Hp_1,\Hp)}<\frac{\eps}{3}$. Hence,
  \begin{align*}
    \norm{P_nAQ_n-A}_{\Sp_p(\Hp_1,\Hp)}
    \le&\ubr{\norm{P_nAQ_n-P_nA_0Q_n}_{\Sp_p(\Hp_1,\Hp)}}_{\le\norm{P_n}_{L(\Hp_1,\Hp)}\norm{A-A_0}_{\Sp_p(\Hp_1,\Hp)}\norm{Q_n}_{L(\Hp_1,\Hp)}}
    +\norm{P_nA_0Q_n-A_0}_{\Sp_p(\Hp_1,\Hp)}\\
    &+\norm{A_0-A}_{\Sp_p(\Hp_1,\Hp)}
  \end{align*}
  implies that it suffices to show that $\norm{P_nA_0Q_n-A_0}_{\Sp_p(\Hp_1,\Hp)}$ is eventually bounded by $\frac\eps3$. By splitting $\norm{P_nA_0Q_n-A_0}_{\Sp_p(\Hp_1,\Hp)}$ into $\norm{P_nA_0(Q_n-1)}_{\Sp_p(\Hp_1,\Hp)}+\norm{(P_n-1)A_0}_{\Sp_p(\Hp_1,\Hp)}$ and, using that the $\Sp_p(\Hp_1,\Hp)$-norm and $L(\Hp_1,\Hp)$-norm coincide on rank-$1$ operators, we observe
  \begin{align*}
    \norm{P_nA_0(Q_n-1)}_{\Sp_p(\Hp_1,\Hp)}
    \le&\sum_{k=0}^{K-1}\abs{\alpha_k}\norm{\langle(Q_n-1)x_k,\cdot\rangle_{\Hp_1}P_ny_k}_{\Sp_p(\Hp_1,\Hp)}\\
    =&\sum_{k=0}^{K-1}\abs{\alpha_k}\norm{\langle(Q_n-1)x_k,\cdot\rangle_{\Hp_1}P_ny_k}_{L(\Hp_1,\Hp)}\\
    \le&\sum_{k=0}^{K-1}\abs{\alpha_k}\ubr{\norm{(Q_n-1)x_k}_{\Hp_1}}_{\to0}\ubr{\norm{P_ny_k}_{\Hp}}_{\le1}
  \end{align*}
  and 
  \begin{align*}
    \norm{(P_n-1)A_0}_{\Sp_p(\Hp_1,\Hp)}
    \le&\sum_{k=0}^{K-1}\abs{\alpha_k}\norm{\langle x_k,\cdot\rangle_{\Hp_1}(P_n-1)y_k}_{\Sp_p(\Hp_1,\Hp)}\\
    =&\sum_{k=0}^{K-1}\abs{\alpha_k}\norm{\langle x_k,\cdot\rangle_{\Hp_1}(P_n-1)y_k}_{L(\Hp_1,\Hp)}\\
    \le&\sum_{k=0}^{K-1}\abs{\alpha_k}\ubr{\norm{x_k}_{\Hp_1}}_{=1}\ubr{\norm{P_ny_k}_{\Hp}}_{\to0}
  \end{align*}
  which completes the proof.
\end{proof}
Later in the proof of Theorem~\ref{thm:continuum-limit}, we will want to show for two trace-class operators $A$ and $B$ that $\frac{\langle\psi_n,A\psi_n\rangle_{\Hp}}{\langle\psi_n,B\psi_n\rangle_{\Hp}}=\lim_{T\to\infty}\frac{\tr(U_nP_nAQ_n)}{\tr(U_nP_nBQ_n)}\to\lim_{T\to\infty}\frac{\tr(UA)}{\tr(UB)}$ holds where $U_n=\texp\l(-\frac i\hbar\int_0^TP_nH(s)Q_nds\r)$. We will do this by showing $\frac{\langle\psi_n,A\psi_n\rangle_{\Hp}}{\langle\psi_n,B\psi_n\rangle_{\Hp}}\to\frac{\langle\psi,A\psi\rangle_\Hp}{\langle\psi,B\psi\rangle_\Hp}$ and using $\langle\psi,A\psi\rangle_\Hp=\lim_{T\to\infty}\frac{\tr(UA)}{Z}$ where $Z$ is the partition function. Hence, the partition function cancels out and we obtain the required result. However, since the partition function itself formally evaluates to $Z=\tr U$ which is mathematically ill-defined, mathematically rigorous existence of $Z$ can be elusive. Hence, in the trace-class setting of the lemma above, we can prove $\frac{\tr(U_nP_nAQ_n)}{\tr(U_nP_nBQ_n)}\to\frac{\tr(UA)}{\tr(UB)}$ directly which would circumvent the problem of existence of $Z$ since the discretized partition function $\tr U_n$ is well-defined ($U_n$ is simply a matrix). 
\begin{lemma}\label{lemma:convergence_of_traces}
  Let the Hamiltonian $H$ satisfy the conditions necessary for the time-dependent Hille-Yosida theorem (cf. Theorem 5.3.1 in$\ $\cite{pazy}) and $\Hp_1$ such that $H(s)\in\Sp_1(\Hp_1,\Hp)$ for all $s\in[0,T]$, as well as $A,B\in\Sp_1(\Hp)$. Then
  \begin{align*}
    \frac{\tr(U_nP_nAQ_n)}{\tr(U_nP_nBQ_n)}\to\frac{\tr(UA)}{\tr(UB)}.
  \end{align*}
\end{lemma}
\begin{proof}
  Let $A_n:=P_nAQ_n$. In order to prove the assertion, it suffices to show that both numerator and denominator converge separately. Using the lemma above,
  \begin{align*}
    \abs{\tr(U_nA_n)-\tr(UA)}
    \le&\abs{\tr((U_n-U)A_n)}+\abs{\tr(U(A_n-A))}\\
    \le&\norm{U_n-U}_{L(\Hp)}\norm{A_n}_{\Sp_1(\Hp)}+\norm{U}_{L(\Hp)}\norm{A_n-A}_{\Sp_1(\Hp)}
  \end{align*}
  implies that it suffices to show $\norm{U_n-U}_{L(\Hp)}\to0$.

  According to the proof of the time-dependent Hille-Yosida theorem, we can define $s_k:=\frac{kT}{K}$ for $0\le k<K$ ($K\in\nn$), let $S_k$ be the semigroup generated by $-\frac i\hbar H(s_k)$, and define a semigroup $\tilde U_K$ as $\tilde U_K(t-s)=S_k(t-s)$ if $s_k\le s\le t\le s_{k+1}$ and $\tilde U_K(t-s)=S_k(t-s_k)S_{k-1}(T/K)\ldots S_{j+1}(T/K)S_{j}(s_{j+1}-s)$ if $s\in[s_j,s_{j+1}]$ and $t\in[s_k,s_{k+1}]$. Then $\norm{\tilde U_K(T,0)-U}_{L(\Hp)}\to0$.

  Furthermore, the Yosida approximations $H_k:=-\frac i\hbar H(s_k)\lambda_k\l(\lambda_k+\frac i\hbar H(s_k)\r)^{-1}$ generate semigroups $R_k$ and we can define the semigroup $\hat U_K$ as $\hat U_K(t-s)=R_k(t-s)$ if $s_k\le s\le t\le s_{k+1}$ and $\hat U_K(t-s)=R_k(t-s_k)R_{k-1}(T/K)\ldots R_{j+1}(T/K)S_{j}(s_{j+1}-s)$ if $s\in[s_j,s_{j+1}]$ and $t\in[s_k,s_{k+1}]$ to obtain $\norm{\hat U_K(T,0)-\tilde U_K(T,0)}_{L(\Hp)}\to0$.

  Let $\eps\in\rn_{>0}$ and $K_0\in\nn$ such that $\fa K\in\nn_{\ge K_0}:\ \norm{\tilde U_K(T,0)-U}_{L(\Hp)}<\frac\eps5$ and $\norm{\hat U_K(T,0)-\tilde U_K(T,0)}_{L(\Hp)}<\frac\eps5$. Since $\fa k:\ H_k\in\Sp_1(\Hp_1,\Hp)$ and there are only finitely many $k$, $\ex N\in\nn\ \fa n\in\nn_{\ge N}:\ \norm{\hat U_K(T,0)-\hat U_K^n(T,0)}_{L(\Hp)}<\frac\eps5$ where $\hat U_K^n$ is constructed correspondingly to $\hat U_K$ but replacing each $H_k$ with $P_nH_kQ_n$.

  Finally, and possibly increasing $N$, $K$, and the $\lambda_k$, we can find $\tilde U_K^n$ with corresponding generators $P_nH(s_k)Q_n$ such that $\fa K\in\nn_{\ge K_0}:\ \norm{\tilde U_K^n(T,0)-U_n}_{L(\Hp)}<\frac\eps5$ and $\norm{\hat U_K^n(T,0)-\tilde U_K^n(T,0)}_{L(\Hp)}<\frac\eps5$, which yields $\norm{U_n-U}_{L(\Hp)}\to0$ and completes the proof.
\end{proof}

\begin{remark*}
  It may be advantageous to consider the proof of Lemma~\ref{lemma:convergence_of_traces} in terms of time-independent Hamiltonians for simplicity. Using Lemma~\ref{lemma:convergence_in_Sp}, it suffices to show that $(U_n(T))_{n\in\nn}$ converges to $U(T)$ in $L(\Hp)$ where $U$ is the semigroup generated by $-\frac{i}{\hbar}H$ and $U_n$ is the semigroup generated by $H_n:=-\frac{i}{\hbar}P_nHQ_n$. Following the theorem of Hille-Yosida we define the Yosida approximations $H_{(\lambda)}:=\lambda H(\lambda-H)^{-1}$ and $H_{(\lambda),n}:=\lambda H_n(\lambda-H_n)^{-1}$ for real $\lambda$ in the respective resolvent sets. The Yosida approximations are thus bounded operators on $\Hp$ and generate the semigroups $U_{(\lambda)}:=\l(\exp(tH_{(\lambda)})\r)_{t\in\rn_{\ge0}}$ and $U_{(\lambda),n}:=\l(\exp(tH_{(\lambda),n})\r)_{t\in\rn_{\ge0}}$. Then, the crucial step in the proof of Hille-Yosida is to show that $\lambda\to\infty$ implies uniform convergence for $t$ in compact subsets of $\rn_{\ge0}$. In particular, $U_{(\lambda)}(T)\to U(T)$ and $U_{(\lambda),n}(T)\to U_{n}(T)$ holds in $L(\Hp)$.

  Hence, it suffices to show $U_{(\lambda),n}(T)\to U_{(\lambda)}(T)$ in $L(\Hp)$ but these are exponential series of bounded operators. In other words, proving $H_{(\lambda),n}\to H_{(\lambda)}$ in $L(\Hp)$ is sufficient. At this point, Lemma~\ref{lemma:convergence_in_Sp} implies $H_n\to H$ in $L(\Hp_1,\Hp)$, i.e., the assertion follows if we can show that $A\mapsto (\lambda-A)^{-1}$ is continuous with respect to $L(\Hp,\Hp_1)$. For this last step, we use the resolvent identity $(\lambda-A_1)^{-1}-(\lambda-A_2)^{-1}=(\lambda-A_1)^{-1}(A_1-A_2)(\lambda-A_2)^{-1}$ which implies $\norm{R_\lambda(H)-R_\lambda(H_n)}_{L(\Hp,\Hp_1)}\le\norm{R_\lambda(H)}_{L(\Hp,\Hp_1)}\norm{H-H_n}_{L(\Hp_1,\Hp)}\norm{R_\lambda(H_n)}_{L(\Hp,\Hp_1)}$ where $R_\lambda(H):=(\lambda-H)^{-1}$ and $R_\lambda(H_n):=(\lambda-H_n)^{-1}$. Thus, the assertion follows from boundedness of $n\mapsto\norm{R_\lambda(H_n)}_{L(\Hp,\Hp_1)}$ which is a consequence of the Neumann series since $\norm{R_\lambda(H_n)}_{L(\Hp,\Hp_1)}\le\frac{\norm{R_\lambda(H)}_{L(\Hp,\Hp_1)}}{1-q}$ holds for $\norm{H-H_n}_{L(\Hp_1,\Hp)}<q\norm{R_\lambda(H)}_{L(\Hp,\Hp_1)}^{-1}$ with $q\in(0,1)$.

  Since semigroups of time-dependent Hamiltonians are constructed using an analogue of the forward Euler method, there remains only one more limit to consider in the proof of Lemma~\ref{lemma:convergence_of_traces}.
\end{remark*}
  
\begin{example*}
  Since the $H(s)$ are usually unbounded operators in $\Hp$, it may at first seem surprising to ask for the $H(s)$ to be of trace-class. However, this is more an assumption about the topology of $\Hp_1$. To illustrate this, we will explicitly discuss a simple example of this case.

  Consider $l_j:\ [0,2\pi]\to\rn;\ x\mapsto\sin(j^2x)$ for each $j\in\nn$, $L:=\lin\{l_j;\ j\in\nn\}$, as well as the following norms on $L$
  \begin{align*}
    \norm\cdot_{L_2}:&\ L\to\rn;\ f\mapsto\norm f_{L_2([0,2\pi])}\\
    \norm\cdot_{W_{2,0}^1}:&\ L\to\rn;\ f\mapsto\norm{\d f}_{L_2([0,2\pi])}.
  \end{align*}
  Then we define $\Hp:=\oli{L}^{\norm\cdot_{L_2}}$ and $\Wp:=\oli{L}^{\norm\cdot_{W_{2,0}^1}}$. Furthermore, let $H$ be a closed operator on $\Wp$ (e.g., a differential operator) and $\Hp_1:=\oli{L}^{\norm\cdot_H}$ where $\norm f_H^2:=\norm f_{W_{2,0}^1}^2+\norm{Hf}_{W_{2,0}^1}^2$. Then, $H$ is bounded as a map from $\Hp_1$ to $\Wp$ and we want to show that it is trace-class as a map from $\Hp_1$ to $\Hp$.

  To show $H\in\Sp_1(\Hp_1,\Hp)$ we first note that $\Wp$ is compactly embedded in $\Hp$, i.e., we can write the identity $\id:\ \Wp\to\Hp$ in the form
  \begin{align*}
    \id=\sum_{j\in\nn}s_j\langle\cdot,e_j\rangle_\Wp f_j
  \end{align*}
  for some orthonormal basis $(e_j)_{j\in\nn}$ of $\Wp$ and orthonormal basis $(f_j)_{j\in\nn}$ of $\Hp$ where $(s_j)_{j\in\nn}$ is the sequence of singular values. In particular, we obtain $\id\in\Sp_p(\Wp,\Hp)$ if and only if $(s_j)_{j\in\nn}\in\ell_p(\nn)$. Considering the $l_j$, we observe
  \begin{align*}
    \langle l_j,l_k\rangle_{L_2([0,2\pi])}=&
    \begin{cases}
      0&,\ j\ne k\\
      \pi&,\ j=k
    \end{cases}
  \end{align*}
  and
  \begin{align*}
    \langle l_j,l_k\rangle_{W^1_{2,0}([0,2\pi])}=&\langle ij^2l_j,ik^2l_k\rangle_{L_2([0,2\pi])}
    =
    \begin{cases}
      0&,\ j\ne k\\
      j^4\pi&,\ j=k
    \end{cases}.
  \end{align*}
  In other words, we can choose $e_j:=\frac{1}{j^2\sqrt\pi}l_j$ and $f_j:=\frac{1}{\sqrt\pi}l_j$. However, this directly yields
  \begin{align*}
    e_k=\sum_{j\in\nn}s_j\langle e_k,e_j\rangle_\Wp f_j=s_kf_k,
  \end{align*}
  which implies $(s_j)_{j\in\nn}=\l(\frac{1}{j^2}\r)_{j\in\nn}\in\ell_1(\nn)$ and, hence, $\id\in\Sp_1(\Wp,\Hp)$.

  Finally, $H\in L(\Hp_1,\Wp)$ and $\id\in\Sp_1(\Wp,\Hp)$ directly imply $H\in\Sp_1(\Hp_1,\Hp)$.
\end{example*}

\section{Proof of Theorem~\ref{thm:continuum-limit}}\label{sec:proof}
Following Proposition~\ref{prop:continuum-limit} we have already proven
\begin{align*}
  \bra\psi A\ket\psi=&\clim\langle A_\disc\rangle=\clim\bra{\psi_\disc} A_\disc\ket{\psi_\disc}=\clim\bra{\psi_n} A\ket{\psi_n}.
\end{align*}
Hence, it remains to show that $\bra\psi A\ket\psi=\lim_{n\to\infty}\langle\psi_n,A\psi_n\rangle_{\Hp}=\langle A\rangle_{\Gf}(0)$. However, since $\langle A\rangle_{\Gf}(0)$ is given by analytic extension, we cannot directly compute it. Instead, we will prove that the sequence $\l(z\mapsto\langle\psi_n,\Gf(z)A\psi_n\rangle_\Hp\r)_{n\in\nn}$ of meromorphic functions is compactly convergent on an open, connected, dense subset of $\cn$ with limit $\langle A\rangle_{\Gf}=\frac{\zeta(U\Gf A)}{\zeta(U\Gf)}$. Recall that for $\Re(z)<R$, both $U\Gf(z)A$ and $U\Gf(z)$ are of trace-class and
\begin{align*}
  \langle A\rangle_\Gf(z)=\lim_{T\to\infty}\frac{\tr(U\Gf(z)A)}{\tr(U\Gf(z))}.
\end{align*}
If $U\Gf(z)$ were a unitary generated by a Hamiltonian, then we could directly interpret this as a vacuum expectation value in some QFT. However, that is not the case, but for the discretized system, we can introduce $U_n=\texp\l(-\frac{i}\hbar \int_0^TP_nH(s)Q_nds\r)$ artificially again;
\begin{align*}
  \langle A_\disc\rangle_\Gf(z)
  =&\lim_{T\to\infty}\frac{\tr(U_nP_n(\Gf(z)A)Q_n)}{\tr(U_nP_n\Gf(z)Q_n)}\\
  =&\lim_{T\to\infty}\frac{\tr\l(U_n P_n(\Gf(z)A)Q_n\r)}{\tr U_n}\lim_{T\to\infty}\frac{\tr U_n}{\tr\l(U_n P_n\Gf(z)Q_n\r)}\\
  =&\frac{\langle\psi_n,\Gf(z)A\psi_n\rangle_\Hp}{\langle\psi_n,\Gf(z)\psi_n\rangle_\Hp}.
\end{align*}
Considering numerator and denominator of the right hand side separately, we can directly see why the assumption for $(z\mapsto\norm{\Gf(z)A\psi_n})_{n\in\nn}$ and $(z\mapsto\norm{\Gf(z)\psi_n})_{n\in\nn}$ to be locally bounded in $C(\cn)$ is necessary. Pointwise boundedness is necessary for both numerator and denominator to be convergent by Proposition~\ref{prop:continuum-limit}. However, pointwise convergence is not quite enough since we need compact convergence and by Vitali's theorem that requires local boundedness.
\begin{theorem}[Vitali's theorem]
  Let $\Omega\sse\cn$ be open and connected and $(f_n)_{n\in\nn}$ a locally bounded sequence of holomorphic functions on $\Omega$ such that $\{z\in \Omega;\ \lim_{n\to\infty}f_n(z)\text{ exists}\}$ has an accumulation point in $\Omega$. Then $(f_n)_{n\in\nn}$ is compactly convergent. 
\end{theorem}
Furthermore, by the following two lemmas, the quotient $z\mapsto\frac{\langle\psi_n,\Gf(z)A\psi_n\rangle_\Hp}{\langle\psi_n,\Gf(z)\psi_n\rangle_\Hp}$ is compactly convergent on an open, dense, connected subset of $\cn$ and coincides with $\langle A\rangle_\Gf$.

\begin{lemma}
  Let $\fa z\in\cn:\ \Dp(z):=\langle\psi,\Gf(z)\psi\rangle_\Hp$. Then $\Dp$ has only isolated zeros. 
\end{lemma}
\begin{proof}
  By Proposition~\ref{prop:continuum-limit}, $\Dp$ is the pointwise limit of the sequence $\l(z\mapsto\langle\psi_n,\Gf(z)\psi_n\rangle_\Hp\r)_{n\in\nn}$ in the denominator and, since the denominator sequence is compactly convergent, $\Dp$ is holomorphic as well. Hence, $\Dp$ has only isolated zeros or $\Dp=0$ (identity theorem for holomorphic functions implies $\Dp=0$ if $[\{0\}]\Dp:=\{z\in\cn;\ \Dp(z)=0\}$ has an accumulation point). However, $\Gf(0)=1$ implies $\Dp(0)=\langle\psi,\psi\rangle_\Hp=1$ and thus the assertion.
\end{proof}

\begin{lemma}
  $\l(z\mapsto\frac{\langle\psi_n,\Gf(z)A\psi_n\rangle_\Hp}{\langle\psi_n,\Gf(z)\psi_n\rangle_\Hp}\r)_{n\in\nn}$ is compactly convergent to $\langle A\rangle_\Gf$ on $\cn\setminus[\{0\}]\Dp$.
\end{lemma}
\begin{proof}
  We obtain pointwise convergence on $\cn\setminus[\{0\}]\Dp$ directly since both numerator and denominator are compactly convergent and the denominator converges pointwise to $\Dp$. Compact convergence thus follows from Vitali's theorem once we have proven local boundedness of the sequence.

  Let $z_0\in\cn\setminus[\{0\}]\Dp$. Then there exists $\delta\in\rn_{>0}$ such that
  \begin{align*}
    \fa z\in B(z_0,\delta):\ \abs{\Dp(z)-\Dp(z_0)}<\frac{\abs{\Dp(z_0)}}{3},
  \end{align*}
  where $B(z_0,\delta):=\{z\in\cn;\ \abs{z-z_0}<\delta\}$, and given such $\delta\in\rn_{>0}$,
  \begin{align*}
    \ex N\in\nn\ \fa n\in\nn_{\ge N}\ \fa z\in B(z_0,\delta):\ \abs{\langle\psi_n,\Gf(z)\psi_n\rangle_\Hp-\Dp(z)}<\frac{\abs{\Dp(z_0)}}{3}.
  \end{align*}
  Hence, $\l(z\mapsto\langle\psi_n,\Gf(z)\psi_n\rangle_\Hp\r)_{n\in\nn}$ is locally eventually bounded away from zero on $\cn\setminus[\{0\}]\Dp$ and therefore $\l(z\mapsto\frac{\langle\psi_n,\Gf(z)A\psi_n\rangle_\Hp}{\langle\psi_n,\Gf(z)\psi_n\rangle_\Hp}\r)_{n\in\nn}$ locally eventually bounded on $\cn\setminus[\{0\}]\Dp$.

  Finally, we need to show that the compact limit is indeed $\langle A\rangle_\Gf$. Since we know that both $\langle A\rangle_\Gf$ and the compact limit of $\l(z\mapsto\frac{\langle\psi_n,\Gf(z)A\psi_n\rangle_\Hp}{\langle\psi_n,\Gf(z)\psi_n\rangle_\Hp}\r)_{n\in\nn}$ are holomorphic on an open, connected, and dense subset $\Omega$ of $\cn$, it suffices to show that $\l(z\mapsto\frac{\langle\psi_n,\Gf(z)A\psi_n\rangle_\Hp}{\langle\psi_n,\Gf(z)\psi_n\rangle_\Hp}\r)_{n\in\nn}$ converges pointwise to $\langle A\rangle_\Gf$ on a set that has an accumulation point in $\Omega$.

  Let $z\in\cn_{\Re\cdot<R}\setminus[\{0\}]\Dp$. Then
  \begin{align*}
    \lim_{n\to\infty}\frac{\langle\psi_n,\Gf(z)A\psi_n\rangle_\Hp}{\langle\psi_n,\Gf(z)\psi_n\rangle_\Hp}=&\frac{\langle\psi,\Gf(z)A\psi\rangle_{\Hp}}{\langle\psi,\Gf(z)\psi\rangle_{\Hp}}
    =\lim_{T\to\infty}\frac{\tr(U\Gf(z)A)}{\tr(U\Gf(z))}
    =\langle A\rangle_\Gf(z)
  \end{align*}
  completes the proof.
\end{proof}

Finally, since $0\notin[\{0\}]\Dp$, we can point evaluate $\l(z\mapsto\frac{\langle\psi_n,\Gf(z)A\psi_n\rangle_\Hp}{\langle\psi_n,\Gf(z)\psi_n\rangle_\Hp}\r)_{n\in\nn}$ and obtain 
\begin{align*}
  \lim_{n\to\infty}\langle\psi_n,A\psi_n\rangle_\Hp=\lim_{n\to\infty}\frac{\langle\psi_n,\Gf(0)A\psi_n\rangle_\Hp}{\langle\psi_n,\Gf(0)\psi_n\rangle_\Hp}=\langle A\rangle_\Gf(0)
\end{align*}
since $\Gf(0)=1$ which completes the proof of Theorem~\ref{thm:continuum-limit}

\section{Conclusion}
It has previously been observed\cite{hartung,hartung-iwota} that $\zeta$-regularization can be applied to Feynman's path integral with Lorentzian background. While being able to obtain physically correct vacuum expectation values in a number of different examples, it was unclear how physical this $\zeta$-regularized path integral is. On the other hand, the $\zeta$-regularized vacuum expectation values are highly interesting because they allow for non-perturbative computations in the continuum with Lorentzian background.

In Theorem~\ref{thm:continuum-limit} we have provided a proof that such $\zeta$-regularized vacuum expectation values $\langle A\rangle_\Gf(0)$ are in fact continuum limits and coincide with the ``true'' vacuum expectation values $\langle A\rangle$ provided the Hamiltonian satisfies certain assumptions. Most of these assumptions are relatively non-critical and can be addressed in a physically meaningful way, e.g., superselection sectors for non-degenerate vacua. Hence, the only assumption of Theorem~\ref{thm:continuum-limit} that is physically relevant is the assumption of an energy gap. As such Theorem~\ref{thm:continuum-limit} is generally applicable to generic quantum field theories. We have shown that at the example of the free radiation field for which we explicitly computed the ground state energy using the $\zeta$-regularized vacuum expectation values in the $N\to\infty$ photon limit (Section~\ref{sec:radiation}) and the full Fock space (Appendix~\ref{app:fockspace}).

The continuum limits used to prove $\langle A\rangle_\Gf(0)=\langle A\rangle$ can be expressed in terms of a discretization scheme $\disc$ which we described in Section~\ref{sec:disc}. This discretization scheme has a couple of remarkable properties, as well. On one hand, the discretized system still has a Lorentzian background, i.e., real time computations are possible. On the other hand, the discrete approximations of the vacuum are accessible on quantum computers where quantum computing is necessary precisely because we are working on a Lorentzian background. In fact, for large, non-trivial systems it is expected that quantum computations are the only way to obtain these expectation values. Hence, the $\zeta$-regularized vacuum expectation values are accessible using quantum computing, too. We have tested this computation of the continuum limit using the Rigetti Quantum Virtual Machine and the Rigetti 8Q~chip ``Agave''.

As an example, we implemented a version of the hydrogen atom in the spatially bounded universe $(-\pi,\pi)$. This forces the ground state to be highly non-trivial from the point of view of the discretization scheme. Nonetheless, the rate of convergence appears to be exponential in the number of qubits ``$\mathrm{error}\sim\exp(-1.92\cdot \mathrm{number\ of\ qubits})$''. We have shown this rate of convergence theoretically (through exact diagonalization of the Hamiltonian) and were able to reproduce it with a quantum computer. A simulation on the Rigetti Quantum Virtual Machine (ignoring noise) agreed with the theoretical values to more than $8$ significant digits using up to $5$ qubits, and an implementation using the Rigetti 8Q~chip ``Agave'' reproduced the theoretical results with the expected accuracy on one qubit. Two or more qubit computations are shown to require quantum processing units with improved fidelity.

Finally, it is important to note that lattice field theories can be expressed as a special case of the discretization scheme considered here, provided that coarser lattices are contained in finer lattices, that is, the coarser lattice spacing is an integer multiple of the finer lattice spacing. The basis used to construct the discretization scheme - which we chose to be the Fourier basis for the $1$-dimensional hydrogen atom - is a basis of piecewise linear functions in a lattice field theory. This is particularly interesting since it was also observed\cite{hartung,hartung-iwota} that lattice field theories can be $\zeta$-regularized in the same way. In other words, the methods described here open up the possibility to study lattice systems on a Lorentzian background. 

\begin{acknowledgments}
  The authors would like to express their gratitude Emilio Fedele and Nazar Miheisi for inspiring comments and conversations which helped to develop the work presented in this article. We would also like to thank the referee for many critical comments which helped improve the article. Furthermore, we sincerely thank Rigetti for access to the Rigetti Quantum Virtual Machine and the Rigetti 8Q quantum chip ``Agave'', as well as, Ryan Karle, Nick Rubin, and Nik Tezak for their support in implementing and running the code on the Rigetti hardware.
\end{acknowledgments}


\appendix
\section{Construction of the Hilbert space and Hamiltonian for the free radiation field}\label{app:radiation-construction}

Considering the free radiation field on $(\rn/X\zn)^3$, we have a gauge field $A$ and the electromagnetic field tensor $F:=dA$, i.e., $F_{\mu\nu}=\d_\mu A_\nu-\d_\nu A_\mu$. Choosing the Lorentz gauge $\d_\mu A^\mu=0$ (for a more detailed expose of the $\rn^3$ case see$\ $~\cite{tong}), the Lagrangian is given by
\begin{align*}
  \Lp=-\frac14 F_{\mu\nu}F^{\mu\nu}-\frac12(\d_\mu A^\mu)^2
\end{align*}
which generates the equation of motion $\d_\mu\d^\mu A=0$.

Turning the classical field $A$ and its canonical momentum~$\pi$, which satisfies $\pi^0=-\d_\mu A^\mu$ and $\pi^j=\d^jA^0-\d^0A^j$, into operators, we impose the canonical commutation relations $[A_\mu(x),A_\nu(x)]=0$, $[\pi^\mu(x),\pi^\nu(x)]=0$, as well as $[A_\mu(x),\pi_\nu(y)]=i\eta_{\mu\nu}\delta^{(3)}(x-y)$. We can now write down $A$ and $\pi$ in terms of creation and annihilation operators ($a^\lambda(p)^\dagger$ and $a^\lambda(p)$)
\begin{widetext}
  \begin{align*}
    A_\mu(x)=&\sum_{p\in\zn^3\setminus\{0\}}\sqrt{\frac{X}{2\pi\norm p_{\ell_2(3)}}}\sum_{\lambda=0}^3\eps_\mu^\lambda(p)\l(a^\lambda(p)e^{\frac{2\pi i}{X}\langle p,x\rangle_{\ell_2(3)}}+a^\lambda(p)^\dagger e^{-\frac{2\pi i}{X}\langle p,x\rangle_{\ell_2(3)}}\r)\\
    \pi^\mu(x)=&\sum_{p\in\zn^3}i\sqrt{\frac{2\pi\norm p_{\ell_2(3)}}{X}}\sum_{\lambda=0}^3(\eps^\mu)^\lambda(p)\l(a^\lambda(p)e^{\frac{2\pi i}{X}\langle p,x\rangle_{\ell_2(3)}}-a^\lambda(p)^\dagger e^{-\frac{2\pi i}{X}\langle p,x\rangle_{\ell_2(3)}}\r)
  \end{align*}
\end{widetext}
where the $4$-vectors $\eps^\lambda$ are the four polarization vectors. Furthermore, we should note that the momentum is endowed with a $+i$ rather than the familiar $-i$ which is due to the Heisenberg picture in which $\pi^\mu=-\d^0A^\mu+\ldots$ generates a factor $+i$. As for the polarization vectors, we choose $\eps^0$ to be timelike and $\eps^1$, $\eps^2$, and $\eps^3$ spacelike with $\eps^3$ longitudinal and $\eps^1$ and $\eps^2$ transversal, i.e.,
\begin{align*}
  \forall \lambda\in\{1,2\}:\ \eps^\lambda_\mu\vec p^\mu=0
\end{align*}
where $\vec p=(\norm p_{\ell_2(3)},p)$ is the photon $4$-momentum. In other words, for momenta $p\propto(1,0,0,1)$, $(\eps^\lambda)_{0\le\lambda<4}$ can be chosen to be the canonical basis of $\rn^4$ and all other polarizations arise applying the appropriate Lorentz transform.

Now we can translate the commutation relations and obtain 
\begin{align*}
  [a^\lambda(p),a^{\lambda'}(q)]=[a^\lambda(p)^\dagger,a^{\lambda'}(q)^\dagger]=0\text{ and }[a^\lambda(p),a^{\lambda'}(q)^\dagger]=-\eta^{\lambda,\lambda'}\delta_{p,q}
\end{align*}
which is fine for spacelike $\lambda$s but 
\begin{align*}
  [a^0(p),a^{0}(q)^\dagger]=-\delta_{p,q}
\end{align*}
is problematic. Since the Lorentz invariant vacuum $\ket0$ is defined via 
\begin{align*}
  \fa 0\le\lambda< 4\ \fa p\in\zn^3:\ a^\lambda(p)\ket0=0,
\end{align*}
we can generate one-photon states $\ket{p,\lambda}=a^\lambda(p)^\dagger\ket0$ and observe
\begin{align*}
  \langle p,0|q,0\rangle=\bra0a^0(p)a^0(q)^\dagger\ket0=-\delta_{p,q},
\end{align*}
that is, $\ket{p,0}$ has negative norm. This is possible because we haven't yet incorporated the Lorentz gauge $\d_\mu A^\mu=0$. To do so, we need to decompose $A_\mu(x)$ into $A_\mu^+(x)+A_\mu^-(x)$ where
\begin{align*}
  A_\mu^+(x)=&\sum_{p\in\zn^3\setminus\{0\}}\sqrt{\frac{X}{2\pi\norm p}}\sum_{\lambda=0}^3\eps_\mu^\lambda(p)a^\lambda(p)e^{\frac{2\pi i}{X}\langle p,x\rangle}\\
  A_\mu^-(x)=&\sum_{p\in\zn^3\setminus\{0\}}\sqrt{\frac{X}{2\pi\norm p}}\sum_{\lambda=0}^3\eps_\mu^\lambda(p)a^\lambda(p)^\dagger e^{-\frac{2\pi i}{X}\langle p,x\rangle}.
\end{align*}
Then a state $\ket\psi$ is physical if and only if 
\begin{align*}
  \d^\mu A_\mu^+\ket\psi=0
\end{align*}
as this ensures
\begin{align*}
  \fa \phi,\psi\text{ physical}:\ \bra\phi\d^\mu A_\mu\ket\psi=0.
\end{align*}
This condition is known as Gupta-Bleuler condition.

However, we still do not have a Hilbert space since we have a degenerate vacuum which means that we only have a semi-inner product. Consider a Fock space basis of the form $\ket{\psi_T}\ket\phi$ where $\ket{\psi_T}$ contains the transversal photons and $\ket\phi$ the timelike and longitudinal photons. The Gupta-Bleuler condition then implies $(a^3(p)-a^0(p))\ket\phi=0$. In other words, any state that contains a timelike photon of momentum $p$ also contains a longitudinal photon of momentum $p$. It is also easy to verify that for any states $\ket{\phi_m}$ with $m$ pairs of timelike and longitudinal photons and $\ket{\phi_n}$ with $n$ pairs of timelike and longitudinal photons $\langle\phi_m|\phi_n\rangle=\delta_{m0}\delta_{n0}$ holds. Taking the quotient with respect to all norm-zero states, we need to make sure that all physical operators $A$ have the same expectation $\bra\phi A\ket\phi$ with respect to all norm-zero states $\ket\phi$.

For the Hamiltonian
\begin{align*}
  H_0=\sum_{p\in\zn^3\setminus\{0\}}\frac{2\pi\norm p}{X}\l(-a^0(p)^\dagger a^0(p)+\sum_{j=1}^3a^j(p)^\dagger a^j(p)\r)
\end{align*}
we can directly check this since the $a^0(p)\ket\phi=a^3(p)\ket\phi$ implies 
\begin{align*}
  \bra{\psi_T,\phi}a^0(p)^\dagger a^0(p)\ket{\psi_T,\phi}=\bra{\psi_T,\phi}a^3(p)^\dagger a^3(p)\ket{\psi_T,\phi}
\end{align*}
and, hence,
\begin{align*}
  \bra{\psi_T,\phi}H_0\ket{\psi_T,\phi}
  =&\bra{\psi_T,\phi}\sum_{p\in\zn^3\setminus\{0\}}\frac{2\pi\norm p_{\ell_2(3)}}{X}\sum_{j=1}^2a^j(p)^\dagger a^j(p)\ket{\psi_T,\phi}\\
  =&\bra{\psi_T}\ubr{\sum_{p\in\zn^3\setminus\{0\}}\frac{2\pi\norm p_{\ell_2(3)}}{X}\sum_{j=1}^2a^j(p)^\dagger a^j(p)}_{=:H}\ket{\psi_T}
\end{align*}
which is independent of $\phi$. More generally, $\bra{\psi_T,\phi}G_0\ket{\psi_T,\phi}=\bra{\psi_T}G\ket{\psi_T}$ can be checked to be true for any gauge-invariant operator $G_0$. In other words, the single photon Hilbert space $\Hp_1$ is spanned by the transversal photon states $\ket{\psi_T}$, i.e.,
\begin{align*}
  \Hp_1=\ell_2\l(\zn^3,\cn^2\r)\ominus\lin\{\ket{0,2}\}
\end{align*}
with basis $(\psi_{p,j})_{p\in\zn^3,j\in\{1,2\}}$ defined as
\begin{align*}
  \fa p\in\zn^3\ \fa j\in\{1,2\}:\ \psi_{p,j}:=\ket{p,j}=a^j(p)^\dagger\ket0,
\end{align*}
and the Hamiltonian is given by
\begin{align*}
  H=\sum_{p\in\zn^3}\frac{2\pi\norm p_{\ell_2(3)}}{X}\l(a^1(p)^\dagger a^1(p)+a^2(p)^\dagger a^2(p)\r)
\end{align*}
where 
\begin{align*}
  \fa p,q\in\zn^3\ \fa j,k\in\{1,2\}:\ a^j(p)^\dagger a^j(p)\psi_{q,k} = \delta_{p,q}\delta_{j,k}.
\end{align*}
Alternatively, we may choose a formulation with 
\begin{align*}
  \Hp_1=L_2\l((\rn/X\zn)^3,\cn^2\r)\ominus\lin\l\{
  \begin{pmatrix}
    0\\1
  \end{pmatrix}
  \r\}
\end{align*}
with basis
\begin{align*}
  \psi_{p,0}(x)=&\frac{1}{\sqrt{X^3}}\exp\l(\frac{2\pi i}{X}\langle p,x\rangle_{\ell_2(3)}\r)
  \begin{pmatrix}
    1\\0
  \end{pmatrix}\\
  \psi_{p,1}(x)=&\frac{1}{\sqrt{X^3}}\exp\l(\frac{2\pi i}{X}\langle p,x\rangle_{\ell_2(3)}\r)
  \begin{pmatrix}
    0\\1
  \end{pmatrix}
\end{align*}
for $p\in\zn^3$ (where $\lin\{\psi_{0,1}\}$ is the orthocomplement of $\Hp_1$ in $L_2\l((\rn/X\zn)^3,\cn^2\r)$) and Hamiltonian
\begin{align*}
  H_1=\abs\nabla\id_{\cn^2}.
\end{align*}
Since the Hamiltonian is diagonalized, it is easy to compute the energy of any given state $\psi=\sum_{(p,j)\in\zn^3\times 2\setminus\{(0,1)\}}\alpha_{p,j}\psi_{p,j}$
\begin{align*}
  \langle\psi_{p,j},H_1\psi_{p,j}\rangle_\Hp=\sum_{(p,j)\in\zn^3\times 2\setminus\{(0,1)\}}\abs{\alpha_{p,j}}^2\frac{2\pi\norm p_{\ell_2(3)}}{X}
\end{align*}
which is minimal if and only if $p=0$, i.e., $\ket0=\psi_{0,0}$.

Choosing any increasing sequence of sets $J_n\sse\{\psi_{p,j};\ (p,j)\in\zn^3\times 2\setminus\{(0,1)\}\}$ with $\#J_n=n$, $J_n\sse J_{n+1}$, and $\bigcup_{n\in\nn}J_n=\zn^3\times 2\setminus\{(0,1)\}$, $\psi_{0,0}$ will eventually be contained in each $P_n[\Hp_1]$ trivializing the continuum limit. Furthermore, this is the expected result as the vacuum should not contain any photons. Yet, since we chose to quotient out norm-zero states of the Fock space earlier, this vacuum does reproduce the non-trivial vacuum containing pairs of longitudinal and timelike photons.

In order to consider $N$-photon states $\ket P=\ket{(p_1,j_1),\ldots,(p_N,j_N)}$, we use the $N$-fold symmetric tensor product\footnote{If we were working in a fermionic theory, the symmetric tensor product would have to be replaced with the antisymmetric tensor product.} $S\bigotimes_{k=1}^N\Hp_1$ of $\Hp_1$ and set $\ket P=\ket{p_1,j_1}\otimes\ldots\otimes\ket{p_N,j_N}$. The Hamiltonian in $S\bigotimes_{k=1}^N\Hp_1$ is given by 
\begin{align*}
  H_N=\sum_{k=1}^N\l(\bigotimes_{m=1}^{k-1}\id_{\Hp_1}\r)\otimes H\otimes\l(\bigotimes_{m=k+1}^{N}\id_{\Hp_1}\r).
\end{align*}

\section{The free radiation field in the Fock space}\label{app:fockspace}
In this appendix, we will discuss the changes to Section~\ref{sec:radiation} necessary to discuss the free radiation field in the Fock space. The Fock space $\Hp_F$ is similar to the up-to-$N$-particle space $\Hp_{\le N}$ in the sense that we do not consider all states of up to $N$ particles but the smallest Hilbert space containing all Hilbert spaces with up to $N$ particles. The Fock space $\Hp_F$ thus minimalizes $\Hp_F\supseteq\bigcup_{n\in\nn}\Hp_{\le N}$. 

Given the single particle Hilbert space $\Hp_1$ - which in the case of the free radiation field is $L_2((\rn/X\zn)^3,\cn^2)\setminus\cn^2$, i.e., the up-to-one-particle Hilbert space with the vacuum removed - the (exactly) $N$-particle space $\Hp_N$ is given by the symmetric tensor product
\begin{align*}
  \Hp_N=S\otimes_{j=1}^N\Hp_1
\end{align*}
and the $N$ particle state $\ket P=\ket{P_0,\ldots,P_{N-1}}=\prod_{j=0}^{N-1}a_{P_j}^\dagger\ket0$ is represented by
\begin{align*}
  \ket P=\frac{1}{N!}\sum_{\pi\in S_N}\bigotimes_{j=0}^{N-1}\ket{P_{\pi(j)}}
\end{align*}
where $S_N$ denotes the symmetric group on the set $N$. The full Fock space $\Hp$ is then the Hilbert space complete direct sum of all exactly-$N$-particle Hilbert spaces
\begin{align*}
  \Hp=\bigoplus_{N\in\nn}\Hp_N=\bigoplus_{N\in\nn}S\otimes_{j=1}^N\Hp_1.
\end{align*}
The Hamiltonian keeps each of the $\Hp_N$ invariant and on each $\Hp_N$ is simply the restriction of  
\begin{align*}
  H_N=\sum_{k=1}^N\l(\bigotimes_{m=1}^{k-1}\id_{\Hp_1}\r)\otimes H\otimes\l(\bigotimes_{m=k+1}^{N}\id_{\Hp_1}\r).
\end{align*}
Thus, the Fock space Hamiltonian $H_F$ is given by
\begin{align*}
  H_F=\bigoplus_{N\in\nn}H_N=\diag\l(\l(\sum_{n=1}^{N}\abs{\nabla_{\rn^{3}}}\r)_{N\in\nn}\r).
\end{align*}
If we now na\"ively attempt to extend the computation of the up-to-$N$-particle Hilbert space to the entire Fock space, we observe the following
\begin{widetext}
  \begin{align*}
    &\text{``}\langle H_F\rangle_\Gf(z)\text{''}\\
    =&\lim_{T\to\infty}\frac{\int_{\bigtimes_{N\in\nn}(\rn^3)^N}\sum_{N\in\nn}\tr\l(e^{-iT\sum_{n=1}^N\norm{\xi_n}_{\ell_2(3)}}\sum_{n=1}^N\norm{\xi_n}_{\ell_2(3)}\prod_{m=1}^N\norm{\xi_m}_{\ell_2(3)}^{z}\id_{\cn^2}\r)d\xi}{\int_{\bigtimes_{N\in\nn}(\rn^3)^N}\sum_{N\in\nn}\tr\l(e^{-iT\sum_{n=1}^N\norm{\xi_n}_{\ell_2(3)}}\prod_{m=1}^N\norm{\xi_m}_{\ell_2(3)}^{z}\id_{\cn^2}\r)d\xi}\\
    =&\lim_{T\to\infty}\frac{\sum_{N\in\nn}\sum_{n=1}^N\prod_{m=1}^N\int_{\rn^3}e^{-iT\norm{\xi_m}_{\ell_2(3)}}\norm{\xi_m}_{\ell_2(3)}^{z+\delta_{mn}}d\xi_m}{\sum_{N\in\nn}\prod_{m=1}^N\int_{\rn^3}e^{-iT\norm{\xi_m}_{\ell_2(3)}}\norm{\xi_m}_{\ell_2(3)}^{z}d\xi_m}\\
    =&\lim_{T\to\infty}\frac{\sum_{N\in\nn}N(\vol\d B_{\rn^3})^N\int_{\rn_{>0}}e^{-iTr}r^{z+3}dr\l(\int_{\rn_{>0}}e^{-iTr}r^{z+2}dr\r)^{N-1}}{\sum_{N\in\nn}(\vol\d B_{\rn^3})^N\l(\int_{\rn_{>0}}e^{-iTr}r^{z+2}dr\r)^N}\\
    =&\lim_{T\to\infty}\frac{\sum_{N\in\nn}N(\vol\d B_{\rn^3})^N\Gamma(z+4)\Gamma(z+3)^{N-1}(iT)^{-z-4}(iT)^{-Nz+z-3N+3}}{\sum_{N\in\nn}(\vol\d B_{\rn^3})^N\Gamma(z+3)^N(iT)^{-Nz-3N}}
  \end{align*}
\end{widetext}
which is a problem because neither numerator nor denominator converge for $T\gg1$ and $\Re(z)\ll0$.

In other words, we need another regularizing factor to control the summation with respect to $N$. In this case, we can define
\begin{align*}
  \alpha_N(z):=N^z(\vol\d B_{\rn^3})^z\frac{\Gamma(4)\Gamma(3)^N}{\Gamma(z+4)\Gamma(z+3)^N}(iT)^{Nz}.
\end{align*}
Then $\fa N\in\nn:\ \alpha_N(0)=1$ and 
\begin{widetext}
  \begin{align*}
    &\langle H_F\rangle_\Gf(z)\\
    =&\lim_{T\to\infty}\frac{\int_{\bigtimes_{N\in\nn}(\rn^3)^N}\sum_{N\in\nn}\tr\l(e^{-iT\sum_{n=1}^N\norm{\xi_n}_{\ell_2(3)}}\sum_{n=1}^N\norm{\xi_n}_{\ell_2(3)}\alpha_N(z)\prod_{m=1}^N\norm{\xi_m}^{z}\id_{\cn^2}\r)d\xi}{\int_{\bigtimes_{N\in\nn}(\rn^3)^N}\sum_{N\in\nn}\tr\l(e^{-iT\sum_{n=1}^N\norm{\xi_n}_{\ell_2(3)}}\alpha_N(z)\prod_{m=1}^N\norm{\xi_m}_{\ell_2(3)}^{z}\id_{\cn^2}\r)d\xi}\\
    =&\lim_{T\to\infty}\frac{\sum_{N\in\nn}\alpha_N(z)N(\vol\d B_{\rn^3})^N\Gamma(z+4)\Gamma(z+3)^{N-1}(iT)^{-Nz-3N-1}}{\sum_{N\in\nn}\alpha_N(z)(\vol\d B_{\rn^3})^N\Gamma(z+3)^N(iT)^{-Nz-3N}}\\
    =&\lim_{T\to\infty}\frac{\sum_{N\in\nn}N^{z+1}(\vol\d B_{\rn^3})^{N+z}\Gamma(4)\Gamma(3)^N\Gamma(z+3)^{-1}(iT)^{-3N-1}}{\sum_{N\in\nn}N^z(\vol\d B_{\rn^3})^{N+z}\Gamma(4)\Gamma(3)^N\Gamma(z+4)^{-1}(iT)^{-3N}}
  \end{align*}
\end{widetext}
which does have convergent numerator and denominator for $\Re(z)\ll0$ and the quotient is in $O\l(T^{-1}\r)$, i.e.,
\begin{align*}
  \langle H_F\rangle_\Gf=0
\end{align*}
which coincides with the $N\to\infty$ particle limit computation in Section~\ref{sec:radiation}.

\nocite{*}
\bibliography{zeta-cont-lim}

\begin{thebibliography}{26}%
\makeatletter
\providecommand \@ifxundefined [1]{%
 \@ifx{#1\undefined}
}%
\providecommand \@ifnum [1]{%
 \ifnum #1\expandafter \@firstoftwo
 \else \expandafter \@secondoftwo
 \fi
}%
\providecommand \@ifx [1]{%
 \ifx #1\expandafter \@firstoftwo
 \else \expandafter \@secondoftwo
 \fi
}%
\providecommand \natexlab [1]{#1}%
\providecommand \enquote  [1]{``#1''}%
\providecommand \bibnamefont  [1]{#1}%
\providecommand \bibfnamefont [1]{#1}%
\providecommand \citenamefont [1]{#1}%
\providecommand \href@noop [0]{\@secondoftwo}%
\providecommand \href [0]{\begingroup \@sanitize@url \@href}%
\providecommand \@href[1]{\@@startlink{#1}\@@href}%
\providecommand \@@href[1]{\endgroup#1\@@endlink}%
\providecommand \@sanitize@url [0]{\catcode `\\12\catcode `\$12\catcode
  `\&12\catcode `\#12\catcode `\^12\catcode `\_12\catcode `\%12\relax}%
\providecommand \@@startlink[1]{}%
\providecommand \@@endlink[0]{}%
\providecommand \url  [0]{\begingroup\@sanitize@url \@url }%
\providecommand \@url [1]{\endgroup\@href {#1}{\urlprefix }}%
\providecommand \urlprefix  [0]{URL }%
\providecommand \Eprint [0]{\href }%
\providecommand \doibase [0]{http://dx.doi.org/}%
\providecommand \selectlanguage [0]{\@gobble}%
\providecommand \bibinfo  [0]{\@secondoftwo}%
\providecommand \bibfield  [0]{\@secondoftwo}%
\providecommand \translation [1]{[#1]}%
\providecommand \BibitemOpen [0]{}%
\providecommand \bibitemStop [0]{}%
\providecommand \bibitemNoStop [0]{.\EOS\space}%
\providecommand \EOS [0]{\spacefactor3000\relax}%
\providecommand \BibitemShut  [1]{\csname bibitem#1\endcsname}%
\let\auto@bib@innerbib\@empty
\bibitem [{\citenamefont {{Creutz}}\ and\ \citenamefont
  {{Freedman}}(1981)}]{creutz-freedman}%
  \BibitemOpen
  \bibfield  {author} {\bibinfo {author} {\bibfnamefont {M.}~\bibnamefont
  {{Creutz}}}\ and\ \bibinfo {author} {\bibfnamefont {B.}~\bibnamefont
  {{Freedman}}},\ }\bibfield  {title} {\enquote {\bibinfo {title} {{A
  Statistical Approach to Quantum Mechanics}},}\ }\href@noop {} {\bibfield
  {journal} {\bibinfo  {journal} {Ann. Phys. (NY)}\ }\textbf {\bibinfo {volume}
  {132}},\ \bibinfo {pages} {427--462} (\bibinfo {year} {1981})}\BibitemShut
  {NoStop}%
\bibitem [{\citenamefont {{Feynman}}(1948)}]{feynman}%
  \BibitemOpen
  \bibfield  {author} {\bibinfo {author} {\bibfnamefont {R.~P.}\ \bibnamefont
  {{Feynman}}},\ }\bibfield  {title} {\enquote {\bibinfo {title} {{Space-Time
  Approach to Non-Relativistic Quantum Mechanics}},}\ }\href@noop {} {\bibfield
   {journal} {\bibinfo  {journal} {Rev. Mod. Phys.}\ }\textbf {\bibinfo
  {volume} {20}},\ \bibinfo {pages} {367--387} (\bibinfo {year}
  {1948})}\BibitemShut {NoStop}%
\bibitem [{\citenamefont {{Feynman}}, \citenamefont {{Hibbs}},\ and\
  \citenamefont {{Styer}}(2005)}]{feynman-hibbs-styer}%
  \BibitemOpen
  \bibfield  {author} {\bibinfo {author} {\bibfnamefont {R.~P.}\ \bibnamefont
  {{Feynman}}}, \bibinfo {author} {\bibfnamefont {A.~R.}\ \bibnamefont
  {{Hibbs}}}, \ and\ \bibinfo {author} {\bibfnamefont {D.~F.}\ \bibnamefont
  {{Styer}}},\ }\href@noop {} {\emph {\bibinfo {title} {{Quantum Mechanics and
  Path Integrals}}}}\ (\bibinfo  {publisher} {Dover Publications, Inc.},\
  \bibinfo {year} {2005})\BibitemShut {NoStop}%
\bibitem [{\citenamefont {{Hawking}}(1977)}]{hawking}%
  \BibitemOpen
  \bibfield  {author} {\bibinfo {author} {\bibfnamefont {S.~W.}\ \bibnamefont
  {{Hawking}}},\ }\bibfield  {title} {\enquote {\bibinfo {title} {{Zeta
  Function Regularization of Path Integrals in Curved Spacetime}},}\
  }\href@noop {} {\bibfield  {journal} {\bibinfo  {journal} {Commun. Math.
  Phys.}\ }\textbf {\bibinfo {volume} {55}},\ \bibinfo {pages} {133--148}
  (\bibinfo {year} {1977})}\BibitemShut {NoStop}%
\bibitem [{\citenamefont {{Kontsevich}}\ and\ \citenamefont
  {{Vishik}}(1994{\natexlab{a}})}]{kontsevich-vishik}%
  \BibitemOpen
  \bibfield  {author} {\bibinfo {author} {\bibfnamefont {M.}~\bibnamefont
  {{Kontsevich}}}\ and\ \bibinfo {author} {\bibfnamefont {S.}~\bibnamefont
  {{Vishik}}},\ }\href@noop {} {\enquote {\bibinfo {title} {{Determinants of
  elliptic pseudo-differential operators}},}\ } (\bibinfo {year}
  {1994}{\natexlab{a}}),\ \bibinfo {note} {max Planck Preprint,
  arXiv:hep-th/9404046}\BibitemShut {NoStop}%
\bibitem [{\citenamefont {{Kontsevich}}\ and\ \citenamefont
  {{Vishik}}(1994{\natexlab{b}})}]{kontsevich-vishik-geometry}%
  \BibitemOpen
  \bibfield  {author} {\bibinfo {author} {\bibfnamefont {M.}~\bibnamefont
  {{Kontsevich}}}\ and\ \bibinfo {author} {\bibfnamefont {S.}~\bibnamefont
  {{Vishik}}},\ }\bibfield  {title} {\enquote {\bibinfo {title} {{Geometry of
  determinants of elliptic operators}},}\ }\href@noop {} {\bibfield  {journal}
  {\bibinfo  {journal} {Functional Analysis on the Eve of the XXI century, Vol.
  I, Progress in Mathematics}\ }\textbf {\bibinfo {volume} {131}},\ \bibinfo
  {pages} {173--197} (\bibinfo {year} {1994}{\natexlab{b}})}\BibitemShut
  {NoStop}%
\bibitem [{\citenamefont {{Ray}}(1970)}]{ray}%
  \BibitemOpen
  \bibfield  {author} {\bibinfo {author} {\bibfnamefont {D.~B.}\ \bibnamefont
  {{Ray}}},\ }\bibfield  {title} {\enquote {\bibinfo {title} {{Reidemeister
  torsion and the Laplacian on lens spaces}},}\ }\href@noop {} {\bibfield
  {journal} {\bibinfo  {journal} {Adv. Math.}\ }\textbf {\bibinfo {volume}
  {4}},\ \bibinfo {pages} {109--126} (\bibinfo {year} {1970})}\BibitemShut
  {NoStop}%
\bibitem [{\citenamefont {{Ray}}\ and\ \citenamefont
  {{Singer}}(1971)}]{ray-singer}%
  \BibitemOpen
  \bibfield  {author} {\bibinfo {author} {\bibfnamefont {D.~B.}\ \bibnamefont
  {{Ray}}}\ and\ \bibinfo {author} {\bibfnamefont {I.~M.}\ \bibnamefont
  {{Singer}}},\ }\bibfield  {title} {\enquote {\bibinfo {title} {{$R$-torsion
  and the Laplacian on Riemannian manifolds}},}\ }\href@noop {} {\bibfield
  {journal} {\bibinfo  {journal} {Adv. Math.}\ }\textbf {\bibinfo {volume}
  {7}},\ \bibinfo {pages} {145--210} (\bibinfo {year} {1971})}\BibitemShut
  {NoStop}%
\bibitem [{\citenamefont {{Hartung}}(2017)}]{hartung}%
  \BibitemOpen
  \bibfield  {author} {\bibinfo {author} {\bibfnamefont {T.}~\bibnamefont
  {{Hartung}}},\ }\bibfield  {title} {\enquote {\bibinfo {title} {{Regularizing
  Feynman Path Integrals using the generalized Kontsevich-Vishik trace}},}\
  }\href@noop {} {\bibfield  {journal} {\bibinfo  {journal} {J. Math. Phys.}\
  }\textbf {\bibinfo {volume} {58}},\ \bibinfo {pages} {123505} (\bibinfo
  {year} {2017})}\BibitemShut {NoStop}%
\bibitem [{\citenamefont {{Hartung}}(2015)}]{hartung-phd}%
  \BibitemOpen
  \bibfield  {author} {\bibinfo {author} {\bibfnamefont {T.}~\bibnamefont
  {{Hartung}}},\ }\emph {\bibinfo {title} {{$\zeta$-functions of Fourier
  Integral Operators}}},\ \href@noop {} {Ph.D. thesis},\ \bibinfo  {school}
  {King's College London} (\bibinfo {year} {2015})\BibitemShut {NoStop}%
\bibitem [{\citenamefont {{Hartung}}(2018)}]{hartung-iwota}%
  \BibitemOpen
  \bibfield  {author} {\bibinfo {author} {\bibfnamefont {T.}~\bibnamefont
  {{Hartung}}},\ }\bibfield  {title} {\enquote {\bibinfo {title} {{Feynman path
  integral regularization using Fourier Integral Operator
  $\zeta$-functions}},}\ }\href@noop {} {\bibfield  {journal} {\bibinfo
  {journal} {The Diversity and Beauty of Applied Operator Theory,
  Birkh{\"a}user}\ ,\ \bibinfo {pages} {261--289}} (\bibinfo {year}
  {2018})}\BibitemShut {NoStop}%
\bibitem [{\citenamefont {{Hartung}}\ and\ \citenamefont
  {{Scott}}(2015)}]{hartung-scott}%
  \BibitemOpen
  \bibfield  {author} {\bibinfo {author} {\bibfnamefont {T.}~\bibnamefont
  {{Hartung}}}\ and\ \bibinfo {author} {\bibfnamefont {S.}~\bibnamefont
  {{Scott}}},\ }\href@noop {} {\enquote {\bibinfo {title} {{A generalized
  Kontsevich-Vishik trace for Fourier Integral Operators and the Laurent
  expansion of $\zeta$-functions}},}\ } (\bibinfo {year} {2015}),\ \bibinfo
  {note} {arXiv:1510.07324v2~[math.AP]}\BibitemShut {NoStop}%
\bibitem [{\citenamefont {{Hartung}}\ and\ \citenamefont
  {{Jansen}}(2019)}]{hartung-jansen}%
  \BibitemOpen
  \bibfield  {author} {\bibinfo {author} {\bibfnamefont {T.}~\bibnamefont
  {{Hartung}}}\ and\ \bibinfo {author} {\bibfnamefont {K.}~\bibnamefont
  {{Jansen}}},\ }\bibfield  {title} {\enquote {\bibinfo {title} {{Integrating
  Gauge Fields in the $\zeta$-formulation of Feynman’s path integral}},}\
  }\href@noop {} {\bibfield  {journal} {\bibinfo  {journal} {{to appear in a
  volume in the series Applied Numerical and Harmonic Analysis in honor of
  Luigi Rodino}}\ } (\bibinfo {year} {2019})},\ \bibinfo {note}
  {arXiv:1902.09926}\BibitemShut {NoStop}%
\bibitem [{\citenamefont {{Guillemin}}(1993)}]{guillemin}%
  \BibitemOpen
  \bibfield  {author} {\bibinfo {author} {\bibfnamefont {V.}~\bibnamefont
  {{Guillemin}}},\ }\bibfield  {title} {\enquote {\bibinfo {title} {{Gauged
  Lagrangian Distributions}},}\ }\href@noop {} {\bibfield  {journal} {\bibinfo
  {journal} {Adv. Math.}\ }\textbf {\bibinfo {volume} {102}},\ \bibinfo {pages}
  {184--201} (\bibinfo {year} {1993})}\BibitemShut {NoStop}%
\bibitem [{\citenamefont {{Radzikowski}}(1992)}]{radzikowski-phd}%
  \BibitemOpen
  \bibfield  {author} {\bibinfo {author} {\bibfnamefont {M.~J.}\ \bibnamefont
  {{Radzikowski}}},\ }\emph {\bibinfo {title} {The Hadamard condition and
  Kay’s conjecture in (axiomatic) quantum field theory on curved
  space-time}},\ \href@noop {} {Ph.D. thesis},\ \bibinfo  {school} {Princeton
  University} (\bibinfo {year} {1992})\BibitemShut {NoStop}%
\bibitem [{\citenamefont {{Radzikowski}}(1996)}]{radzikowski}%
  \BibitemOpen
  \bibfield  {author} {\bibinfo {author} {\bibfnamefont {M.~J.}\ \bibnamefont
  {{Radzikowski}}},\ }\bibfield  {title} {\enquote {\bibinfo {title}
  {{Micro-local approach to the Hadamard condition in quantum field theory on
  curved space-time}},}\ }\href@noop {} {\bibfield  {journal} {\bibinfo
  {journal} {Commun. Math. Phys.}\ }\textbf {\bibinfo {volume} {179}},\
  \bibinfo {pages} {529--553} (\bibinfo {year} {1996})}\BibitemShut {NoStop}%
\bibitem [{\citenamefont {{Brislawn}}(1988)}]{brislawn}%
  \BibitemOpen
  \bibfield  {author} {\bibinfo {author} {\bibfnamefont {C.}~\bibnamefont
  {{Brislawn}}},\ }\bibfield  {title} {\enquote {\bibinfo {title} {{Kernels of
  trace class operators}},}\ }\href@noop {} {\bibfield  {journal} {\bibinfo
  {journal} {Proc. Am. Math. Soc.}\ }\textbf {\bibinfo {volume} {104}},\
  \bibinfo {pages} {1181--1190} (\bibinfo {year} {1988})}\BibitemShut {NoStop}%
\bibitem [{\citenamefont {{Scott}}(2010)}]{scott}%
  \BibitemOpen
  \bibfield  {author} {\bibinfo {author} {\bibfnamefont {S.}~\bibnamefont
  {{Scott}}},\ }\href@noop {} {\emph {\bibinfo {title} {Traces and determinants
  of pseudodifferential operators}}}\ (\bibinfo  {publisher} {Oxford University
  Press},\ \bibinfo {year} {2010})\BibitemShut {NoStop}%
\bibitem [{\citenamefont {{Peruzzo}}\ \emph {et~al.}(2014)\citenamefont
  {{Peruzzo}}, \citenamefont {{McClean}}, \citenamefont {{Shadbolt}},
  \citenamefont {{Yung}}, \citenamefont {{Zhou}}, \citenamefont {{Love}},
  \citenamefont {{Aspuru-Guzik}},\ and\ \citenamefont {{O'Brien}}}]{vqe}%
  \BibitemOpen
  \bibfield  {author} {\bibinfo {author} {\bibfnamefont {A.}~\bibnamefont
  {{Peruzzo}}}, \bibinfo {author} {\bibfnamefont {J.}~\bibnamefont
  {{McClean}}}, \bibinfo {author} {\bibfnamefont {P.}~\bibnamefont
  {{Shadbolt}}}, \bibinfo {author} {\bibfnamefont {M.~H.}\ \bibnamefont
  {{Yung}}}, \bibinfo {author} {\bibfnamefont {X.~Q.}\ \bibnamefont {{Zhou}}},
  \bibinfo {author} {\bibfnamefont {P.~J.}\ \bibnamefont {{Love}}}, \bibinfo
  {author} {\bibfnamefont {A.}~\bibnamefont {{Aspuru-Guzik}}}, \ and\ \bibinfo
  {author} {\bibfnamefont {J.~L.}\ \bibnamefont {{O'Brien}}},\ }\bibfield
  {title} {\enquote {\bibinfo {title} {{A variational eigenvalue solver on a
  photonic quantum processor}},}\ }\href@noop {} {\bibfield  {journal}
  {\bibinfo  {journal} {Nat. Commun.}\ }\textbf {\bibinfo {volume} {5}},\
  \bibinfo {pages} {4213} (\bibinfo {year} {2014})}\BibitemShut {NoStop}%
\bibitem [{\citenamefont {{Rubin}}(2016)}]{rubin}%
  \BibitemOpen
  \bibfield  {author} {\bibinfo {author} {\bibfnamefont {N.~C.}\ \bibnamefont
  {{Rubin}}},\ }\href@noop {} {\enquote {\bibinfo {title} {A hybrid
  classical/quantum approach for large-scale studies of quantum systems with
  density matrix embedding theory},}\ } (\bibinfo {year} {2016}),\ \bibinfo
  {note} {arXiv:1610.06910v2~[quant-ph]}\BibitemShut {NoStop}%
\bibitem [{\citenamefont {{H{\"o}rmander}}(1983)}]{hoermander}%
  \BibitemOpen
  \bibfield  {author} {\bibinfo {author} {\bibfnamefont {L.}~\bibnamefont
  {{H{\"o}rmander}}},\ }\href@noop {} {\emph {\bibinfo {title} {{The analysis
  of linear partial differential operators I.}}}}\ (\bibinfo  {publisher}
  {Grundl. math. Wiss. 256. Springer-Verlag},\ \bibinfo {year}
  {1983})\BibitemShut {NoStop}%
\bibitem [{\citenamefont {{Smith}}(2017)}]{smith}%
  \BibitemOpen
  \bibfield  {author} {\bibinfo {author} {\bibfnamefont {R.~S.}\ \bibnamefont
  {{Smith}}},\ }\href@noop {} {\enquote {\bibinfo {title} {Someone shouts,
  ``$|01000\rangle$!'' who is excited?}}\ } (\bibinfo {year} {2017}),\ \bibinfo
  {note} {arXiv:1711.02086v1~[quant-ph]}\BibitemShut {NoStop}%
\bibitem [{\citenamefont {{Smith}}, \citenamefont {{Curtis}},\ and\
  \citenamefont {{Zeng}}()}]{quil}%
  \BibitemOpen
  \bibfield  {author} {\bibinfo {author} {\bibfnamefont {R.~S.}\ \bibnamefont
  {{Smith}}}, \bibinfo {author} {\bibfnamefont {M.~J.}\ \bibnamefont
  {{Curtis}}}, \ and\ \bibinfo {author} {\bibfnamefont {W.~J.}\ \bibnamefont
  {{Zeng}}},\ }\href@noop {} {\enquote {\bibinfo {title} {A practical quantum
  instruction set architecture},}\ }\BibitemShut {NoStop}%
\bibitem [{\citenamefont {{Reagor}}\ \emph {et~al.}()\citenamefont {{Reagor}},
  \citenamefont {{Osborn}}, \citenamefont {{Tezak}}, \citenamefont {{Staley}},
  \citenamefont {{Prawiroatmodjo}}, \citenamefont {{Scheer}}, \citenamefont
  {{Alidoust}}, \citenamefont {{Sete}}, \citenamefont {{Didier}}, \citenamefont
  {{Da Silva}}, \citenamefont {{Acala}}, \citenamefont {{Anegeles}},
  \citenamefont {{Bestwick}}, \citenamefont {{Block}}, \citenamefont {{Bloom}},
  \citenamefont {{Bradley}}, \citenamefont {{Bui}}, \citenamefont {{Caldwell}},
  \citenamefont {{Capelluto}}, \citenamefont {{Chilcott}}, \citenamefont
  {{Cordova}}, \citenamefont {{Crossman}}, \citenamefont {{Curtis}},
  \citenamefont {{Deshpande}}, \citenamefont {{El Bouayadi}}, \citenamefont
  {{Girshovich}}, \citenamefont {{Hong}}, \citenamefont {{Hudson}},
  \citenamefont {{Karalekas}}, \citenamefont {{Kuang}}, \citenamefont
  {{Lenihan}}, \citenamefont {{Manenti}}, \citenamefont {{Manning}},
  \citenamefont {{Marshall}}, \citenamefont {{Mohan}}, \citenamefont
  {{O'Brien}}, \citenamefont {{Otterbach}}, \citenamefont {{Papageorge}},
  \citenamefont {{Paquette}}, \citenamefont {{Pelstring}}, \citenamefont
  {{Polloreno}}, \citenamefont {{Rawat}}, \citenamefont {{Ryan}}, \citenamefont
  {{Renzas}}, \citenamefont {{Rubin}}, \citenamefont {{Russel}}, \citenamefont
  {{Rust}}, \citenamefont {{Scarabelli}}, \citenamefont {{Selvanayagam}},
  \citenamefont {{Sinclair}}, \citenamefont {{Smith}}, \citenamefont {{Suska}},
  \citenamefont {{To}}, \citenamefont {{Vahidpour}}, \citenamefont
  {{Vodrahalli}}, \citenamefont {{Whyland}}, \citenamefont {{Yadav}},
  \citenamefont {{Zeng}},\ and\ \citenamefont {{Rigetti}}}]{rigetti}%
  \BibitemOpen
  \bibfield  {author} {\bibinfo {author} {\bibfnamefont {M.}~\bibnamefont
  {{Reagor}}}, \bibinfo {author} {\bibfnamefont {C.~B.}\ \bibnamefont
  {{Osborn}}}, \bibinfo {author} {\bibfnamefont {N.}~\bibnamefont {{Tezak}}},
  \bibinfo {author} {\bibfnamefont {A.}~\bibnamefont {{Staley}}}, \bibinfo
  {author} {\bibfnamefont {G.}~\bibnamefont {{Prawiroatmodjo}}}, \bibinfo
  {author} {\bibfnamefont {M.}~\bibnamefont {{Scheer}}}, \bibinfo {author}
  {\bibfnamefont {N.}~\bibnamefont {{Alidoust}}}, \bibinfo {author}
  {\bibfnamefont {E.~A.}\ \bibnamefont {{Sete}}}, \bibinfo {author}
  {\bibfnamefont {N.}~\bibnamefont {{Didier}}}, \bibinfo {author}
  {\bibfnamefont {M.~P.}\ \bibnamefont {{Da Silva}}}, \bibinfo {author}
  {\bibfnamefont {E.}~\bibnamefont {{Acala}}}, \bibinfo {author} {\bibfnamefont
  {J.}~\bibnamefont {{Anegeles}}}, \bibinfo {author} {\bibfnamefont
  {A.}~\bibnamefont {{Bestwick}}}, \bibinfo {author} {\bibfnamefont
  {M.}~\bibnamefont {{Block}}}, \bibinfo {author} {\bibfnamefont
  {B.}~\bibnamefont {{Bloom}}}, \bibinfo {author} {\bibfnamefont
  {A.}~\bibnamefont {{Bradley}}}, \bibinfo {author} {\bibfnamefont
  {C.}~\bibnamefont {{Bui}}}, \bibinfo {author} {\bibfnamefont
  {S.}~\bibnamefont {{Caldwell}}}, \bibinfo {author} {\bibfnamefont
  {L.}~\bibnamefont {{Capelluto}}}, \bibinfo {author} {\bibfnamefont
  {R.}~\bibnamefont {{Chilcott}}}, \bibinfo {author} {\bibfnamefont
  {J.}~\bibnamefont {{Cordova}}}, \bibinfo {author} {\bibfnamefont
  {G.}~\bibnamefont {{Crossman}}}, \bibinfo {author} {\bibfnamefont
  {M.}~\bibnamefont {{Curtis}}}, \bibinfo {author} {\bibfnamefont
  {S.}~\bibnamefont {{Deshpande}}}, \bibinfo {author} {\bibfnamefont
  {T.}~\bibnamefont {{El Bouayadi}}}, \bibinfo {author} {\bibfnamefont
  {D.}~\bibnamefont {{Girshovich}}}, \bibinfo {author} {\bibfnamefont
  {S.}~\bibnamefont {{Hong}}}, \bibinfo {author} {\bibfnamefont
  {A.}~\bibnamefont {{Hudson}}}, \bibinfo {author} {\bibfnamefont
  {P.}~\bibnamefont {{Karalekas}}}, \bibinfo {author} {\bibfnamefont
  {K.}~\bibnamefont {{Kuang}}}, \bibinfo {author} {\bibfnamefont
  {M.}~\bibnamefont {{Lenihan}}}, \bibinfo {author} {\bibfnamefont
  {R.}~\bibnamefont {{Manenti}}}, \bibinfo {author} {\bibfnamefont
  {T.}~\bibnamefont {{Manning}}}, \bibinfo {author} {\bibfnamefont
  {J.}~\bibnamefont {{Marshall}}}, \bibinfo {author} {\bibfnamefont
  {Y.}~\bibnamefont {{Mohan}}}, \bibinfo {author} {\bibfnamefont
  {W.}~\bibnamefont {{O'Brien}}}, \bibinfo {author} {\bibfnamefont
  {J.}~\bibnamefont {{Otterbach}}}, \bibinfo {author} {\bibfnamefont
  {A.}~\bibnamefont {{Papageorge}}}, \bibinfo {author} {\bibfnamefont {J.~P.}\
  \bibnamefont {{Paquette}}}, \bibinfo {author} {\bibfnamefont
  {M.}~\bibnamefont {{Pelstring}}}, \bibinfo {author} {\bibfnamefont
  {A.}~\bibnamefont {{Polloreno}}}, \bibinfo {author} {\bibfnamefont
  {V.}~\bibnamefont {{Rawat}}}, \bibinfo {author} {\bibfnamefont {C.~A.}\
  \bibnamefont {{Ryan}}}, \bibinfo {author} {\bibfnamefont {R.}~\bibnamefont
  {{Renzas}}}, \bibinfo {author} {\bibfnamefont {N.}~\bibnamefont {{Rubin}}},
  \bibinfo {author} {\bibfnamefont {D.}~\bibnamefont {{Russel}}}, \bibinfo
  {author} {\bibfnamefont {M.}~\bibnamefont {{Rust}}}, \bibinfo {author}
  {\bibfnamefont {D.}~\bibnamefont {{Scarabelli}}}, \bibinfo {author}
  {\bibfnamefont {M.}~\bibnamefont {{Selvanayagam}}}, \bibinfo {author}
  {\bibfnamefont {R.}~\bibnamefont {{Sinclair}}}, \bibinfo {author}
  {\bibfnamefont {R.}~\bibnamefont {{Smith}}}, \bibinfo {author} {\bibfnamefont
  {M.}~\bibnamefont {{Suska}}}, \bibinfo {author} {\bibfnamefont {T.~W.}\
  \bibnamefont {{To}}}, \bibinfo {author} {\bibfnamefont {M.}~\bibnamefont
  {{Vahidpour}}}, \bibinfo {author} {\bibfnamefont {N.}~\bibnamefont
  {{Vodrahalli}}}, \bibinfo {author} {\bibfnamefont {T.}~\bibnamefont
  {{Whyland}}}, \bibinfo {author} {\bibfnamefont {K.}~\bibnamefont {{Yadav}}},
  \bibinfo {author} {\bibfnamefont {W.}~\bibnamefont {{Zeng}}}, \ and\ \bibinfo
  {author} {\bibfnamefont {C.~T.}\ \bibnamefont {{Rigetti}}},\ }\bibfield
  {title} {\enquote {\bibinfo {title} {Demonstration of universal parametric
  entangling gates on a multi-qubit lattice},}\ }\href@noop {} {\bibfield
  {journal} {\bibinfo  {journal} {Sci. Adv.}\ }\textbf {\bibinfo {volume}
  {4}},\ \bibinfo {pages} {eaao3603}}\BibitemShut {NoStop}%
\bibitem [{\citenamefont {{Pazy}}(1992)}]{pazy}%
  \BibitemOpen
  \bibfield  {author} {\bibinfo {author} {\bibfnamefont {A.}~\bibnamefont
  {{Pazy}}},\ }\href@noop {} {\emph {\bibinfo {title} {Semigroups of Linear
  Operators and Applications to Partial Differential Equations}}}\ (\bibinfo
  {publisher} {Springer},\ \bibinfo {year} {1992})\BibitemShut {NoStop}%
\bibitem [{\citenamefont {{Tong}}(2006)}]{tong}%
  \BibitemOpen
  \bibfield  {author} {\bibinfo {author} {\bibfnamefont {D.}~\bibnamefont
  {{Tong}}},\ }\href@noop {} {\emph {\bibinfo {title} {Quantum Field Theory}}}\
  (\bibinfo  {publisher} {University of Cambridge Part III Mathematical
  Tripos},\ \bibinfo {year} {2006})\ \bibinfo {note}
  {http://www.damtp.cam.ac.uk/user/tong/qft/qft.pdf}\BibitemShut {NoStop}%
\end{thebibliography}%

\end{document}